\PassOptionsToPackage{table}{xcolor}
\documentclass[a4paper,UKenglish,cleveref, autoref, thm-restate]{lipics-v2021}

\hideLIPIcs  

\usepackage{xspace}
\usepackage{diagbox}
\usepackage{makecell}
\usepackage{multirow}
\usepackage{macros}

\newtheorem{property}[theorem]{Property}
\Crefname{property}{Property}{Properties}

\Crefname{section}{Sec.}{Secs.}
\Crefname{figure}{Fig.}{Figs.}

\title{What does it take to certify a conversion checker?} 

\author{Meven Lennon-Bertrand}{University of Cambridge, United Kingdom}{Meven.Lennon-Bertrand@cl.cam.ac.uk}{https://orcid.org/0000-0002-7079-8826}{}

\authorrunning{M. Lennon-Bertrand} 

\Copyright{Meven Lennon-Bertrand} 



\relatedversion{This is the extended version, with complete typing rules in appendix,
  of the following, published at FSCD 2025.} 
\relatedversiondetails[cite=FSCDVersion]{FSCD}{https://doi.org/10.4230/LIPIcs.FSCD.2025.27} 
\supplement{~}
\supplementdetails[subcategory={GitHub}]{Software}{https://github.com/CoqHott/logrel-coq/tree/fscd25}
\supplementdetails[subcategory={Zenodo}]{Software}{https://zenodo.org/records/15359259}  

\funding{
  This work was supported in part by a European Research Council (ERC) Consolidator Grant for the project ``TypeFoundry'', funded under the European Union's Horizon 2020 Framework Programme (grant agreement no. 101002277).}

\acknowledgements{Many thanks to Kenji Maillard and Thibaut Benjamin for their feedback on drafts
of this paper, and to the anonymous reviewers for their helpful comments.}

\nolinenumbers 

\EventEditors{Maribel Fern\'{a}ndez}
\EventNoEds{1}
\EventLongTitle{10th International Conference on Formal Structures for Computation and Deduction (FSCD 2025)}
\EventShortTitle{FSCD 2025}
\EventAcronym{FSCD}
\EventYear{2025}
\EventDate{July 14--20, 2025}
\EventLocation{Birmingham, UK}
\EventLogo{}
\SeriesVolume{337}
\ArticleNo{24}

\begin{document}

\maketitle

\begin{abstract}
We report on a detailed exploration of the properties of conversion (definitional equality) in dependent type theory, with the goal of certifying
decision procedures for it. While in that context the property of normalisation
has attracted the most light, we instead emphasize the importance of \emph{injectivity}
properties, showing that they alone are both crucial and
sufficient to certify most desirable properties of conversion checkers.
We also explore the certification of a fully untyped conversion checker, with
respect to a typed specification, and show that the story is mostly unchanged,
although the exact injectivity properties needed are subtly different.
\end{abstract}

\section{Introduction}
\label{sec:intro}

\subparagraph{Certifying the certifier}

Proof assistants are more and more popular tools, that have now been used to formally
certify many critical software systems and libraries.
The trust in these relies on the proof assistant, which itself
generally concentrates trust on a \emph{kernel}, a well-delimited
and relatively simple piece of the system,
in charge of (re)checking all proofs generated outside it.
This so-called de Bruijn criterion \cite{Barendregt2001}
is a first step toward safety. However, how can we make the kernel, the
keystone of the whole ecosystem, trustworthy?

We should certify the kernel itself! Since it is small and
well specified, this surely should not be hard? In higher order
logic, \textsc{Candle}~\cite{Abrahamsson2022} already
provides such a certified kernel for \textsc{HOL Light}.
An equivalent achievement for dependent type theory is however still missing,
despite 30 years of efforts since Barras \cite{Barras1997}, and the
recent \Agdac~\cite{Liesnikov2025}, \Leanf~\cite{Carneiro2024} and
\MetaCoq~\cite{MetaCoq2024} projects.
The reason is that kernels for dependent type theories,
although relatively simple, rely on many invariants.
Establishing these invariants requires many
properties of the type system that the kernel is supposed to decide.

Yet, showing these properties is a tremendous undertaking,%
\footnote{\MetaCoq has already exceeded the 200kLoC,
and its meta-theory is far from exhaustive.}
when not altogether impossible: due to Gödel's
second incompleteness theorem, any property implying a system's consistency
cannot be proven in the system itself. This applies in particular to
termination of the kernel, which relies on the infamous \emph{normalisation}
property, which in turn implies consistency.
The full certification of a dependently-typed proof assistant
in itself is thus unfeasible. 
But what can be salvaged?
We should aim for graceful failure: even if we cannot reach normalisation,
we would like to certify \emph{something} about our type checker, and identify
the key \emph{meta-theoretic} properties of the type system necessary for this.


By clearly divorcing the meta-theory, taken as an opaque toolbox,
from the certification of the algorithm, we
obtain a good separation of concerns. This should help
in getting bolder on \emph{both} sides: we can modularly adopt advanced proof techniques
on the meta-theory side \cite{Sterling2021a,Bocquet2023},
and independently aim for the certification of complex, optimised implementations on the
certification one \cite{Courant2024,Kovacs2024}.
For this, it is crucial that the meta-theory is not tailored to a particular
implementation, and vice-versa.

Let us also emphasise that, while a fully certified kernel
proves that typing is decidable,
mere decidability is in this view not central: we care about
the actual content of the typing algorithm(s).
This is even clearer in case we cannot ---or do not wish to---
prove normalisation, and therefore decidability.
We still want to talk about a typing algorithm,
which cannot be extracted from a decidability proof that does not exist.



\subparagraph{Injectivity}

As often with dependent types, the main source of meta-theoretic complications
is conversion, the equational theory up to
which types are compared, which we write \(\convop\).%
\footnote{This is a typed judgment, although
we sometimes omit the context and/or type in informal explanations.}
Out of the many properties that can be demanded of it,
we want to put forward \emph{injectivity}
---as well as the accompanying \emph{no-confusion}.
For the \(\P\) type constructor,
injectivity says that if \(\P x : A.B \convop \P x : A'. B'\),
then \(A \convop A'\) and \(B \convop B'\), while no-confusion says,
for instance, that \(\P x : A.B \convop \nat\) is impossible.
These cannot be shown by mere inversion of the conversion derivation,
which can go through a long chain of transitivities.
Injectivity shortcuts this, saying that, despite transitivity,
certain congruences are invertible.
This is a key aspect of intentional type theories, whose conversion is sufficiently
restricted ---injectivity of type constructors
fails in extensional type theories \cite{Winterhalter2020}.


What makes injectivity attractive is that despite its importance
it does not imply logical consistency ---although no-confusion implies
\emph{equational consistency}: not all equations hold. Consequently,
and in contrast with normalisation and its Gödelian limitations,
it can be proven in a meta-theory weaker than the object theory
\cite{Takahashi1995,MetaCoq2024,Coquand2018a}.

\subparagraph{Untyped conversion}
To better understand the relation between meta-theory and
checker, we actually explore two conversion checking algorithms.
The first is a type-directed one, similar to that already
certified by Adjedj \etal\ \cite{Adjedj2024}, and Abel \etal\ \cite{Abel2017} before them.
Our main addition here is to more finely decompose this certification.
However, this algorithm is not really faithful to the kernels
of proof assistants: in \Coq and \Lean, conversion does not
maintain any type information, and 
\Agda's conversion, while primarily type-directed,
similarly uses term-directed \(\eta\)-expansion of functions.
We thus attack the novel certification of an untyped conversion checker,
with \(\eta\)-laws.

Importantly, the specification remains typed:
\emph{specifying} \(\eta\)-laws in an untyped way is
perilous \cite{LennonBertrand2022a}, although their \emph{implementation} is reasonable
\cite{Coquand1996}, as we verify.
An important takeaway of our work is that, in a sense, talking about untyped conversion
is misleading: we should rather be talking about \emph{term-directed typed conversion}.
We are not the first to try and relate
typed and untyped conversion \cite{Siles2012,Abel2007}, although we
hope to demonstrate that concentrating on algorithms makes the proof of equivalence
pleasantly straightforward compared to those previous works.

\subparagraph{Certified decision procedures}

We have introduced the main meta-theoretic properties we focus on:
normalisation and injectivity.
On the other side, what do we want to certify of the type/conversion checker?
Abstractly, these take some input data, and hopefully return
either a positive answer ---and possibly some data associated to it, \eg an
inferred type--- or a negative answer. 
We say ``hopefully'' because conversion checking is non-structurally
recursive, and thus a priori partial.
Given such a potentially partial procedure \(p\), supposed to correspond
to a proposition \(P\), we can decompose its correctness into
\begin{itemize}
  \item \emph{positive soundness}: if \(p\) answers positively, then \(P\) holds;
  \item \emph{negative soundness}: if \(p\) answers negatively, then \(\neg P\)
    holds;
  \item \emph{termination}: \(p\) always answers in one of these two ways.
\end{itemize}
Together, they imply that \(p\) is a decision procedure for \(P\), in
particular we have
\emph{completeness}: if \(P\) holds then \(p\) answers positively.
%
We show that injectivity of type constructors is
enough to prove positive soundness; with extra injectivity properties for terms,
which differ slightly between the two conversion-checking algorithms,
we establish negative soundness;
instead, adding normalisation, we obtain termination.

\begin{figure}
  \begingroup
  \renewcommand{\arraystretch}{2.5}
  \centering
  \begin{tabular}{r c c c c}
    & \makecell{Positive\\soundness} & \makecell{Negative soundness\\(typed conversion)}
    & \makecell{Negative soundness\\(untyped conversion)} &
      Termination\\
    \makecell{Injectivity of\\type constructors} & \(\times\) & \(\times\) & \(\times\) &
    \(\times\) \\
    \makecell{Term-level\\injectivities} & & \(\times\) & \(\times\) & \\
    \makecell{Normalisation} &&&& 
    \(\times\)
  \end{tabular}
  \endgroup

  \caption{Sufficient meta-theoretic properties for certification:
    to show a property of the algorithms (column) it suffices to have the
    properties of the type systems ticked in the corresponding lines.
    For instance, for positive soundness, injectivity of type constructors is sufficient.
  }
  \label{fig:meta-algo-prop}
\end{figure}

\subparagraph{Summary of results}
Building on a previous formalisation \cite{Adjedj2024} in the \Coq proof assistant~\cite{CDT2024}
(previously known as \textsc{Coq}),
we certify the properties gathered in \cref{fig:meta-algo-prop},
for two typing algorithms,
using respectively type-directed and untyped conversion checking,
for a version of Martin-Löf type theory (MLTT).%
\footnote{Featuring dependent function (\(\P\)) and pair (\(\Sig\)) types
with their \(\eta\)-laws, a universe, an empty type,
natural numbers, and an identity type, all with large elimination.}
The rest of this paper is organised as follows:
\begin{itemize}
  \item \cref{sec:declarative} summarises the declarative specification
    and its properties;
  \item \cref{sec:algorithms} presents the algorithms we certify;
  \item \cref{sec:certification} explores their certification
    and the relation of typed and untyped conversion;
  \item \cref{sec:conclusion} concludes with related and future work.
\end{itemize}
Our formalisation adds roughly 10kLoC to the original code base.
Links to it are indicated by this symbol: \formfile[tree/fscd25]{}.
Full versions of inference rules and other definitions are given in the appendix.

\section{The declarative system and its properties}
\label{sec:declarative}

\subsection{Definitions}


  

  


\subparagraph{\formfile{DeclarativeTyping.v} The type system}
Our declarative type system and specification is a version of Martin-Löf's type theory
\cite{MartinLoef1984} with its five mutually defined judgments:
typing of types, terms, and contexts
and conversion (definitional equality) of types and terms, written \(\convop\).
They are defined as mutual inductive predicates on ``raw'' objects.
We use \(=\) to denote meta-level equality ---which coincides
\(\alpha\)-equality, as our representation uses pure de Bruijn indices.

The language supports dependent function (\(\P\)) and pair (\(\Sig\)) types,
which we call ``negative'', with their \(\beta\) and \(\eta\)-rules.
%
We also include inductive types, with dependent large eliminators:
an empty type \(\Empty\), natural number type \(\nat\), and an identity type \(\Id\).
Finally, we have a universe \(\univ\), with codes for all type formers except itself.
We refer to these types as ``positive''.
Although our work does not directly involve focusing or logical polarity, where
the positive/negative distinction stems from, it will still be an important distinction
in the way our conversion algorithms work.
We focus on \(\univ\), \(\P\) and \(\nat\) as representative examples on paper
and direct interested readers to the code for others.


The development relies on parallel substitution,
written \(\sub{t}{\sigma}\).
If \(u\) is a term, we also denote as \(u\) the substitution mapping the
last variable in context to \(u\) and leaving all others untouched,
as used in \eg \(\beta\)-reduction: \((\l x : A.t)~u \red \sub{t}{u}\).
In \Coq, we use \AutoSubst~2~\cite{Stark2020,Dapprich2021}
to generate boilerplate for the substitution calculus
and automate tedious proofs with
a decision procedure for meta-level equality of expressions in this calculus.

\begin{figure}
  \[\begin{array}{l c l r}
    \boxed{\can f} & \eqdef& \univ \mid \P x : t. t \mid  \l x : t. t \mid \nat \mid 0 \mid \succ(t) 
    & \text{\textcolor{gray}{weak-head canonical forms}} \\
    \boxed{\neu n} & \eqdef& x \mid n\ t 
      \mid \natelim{n}{x.t}{t}{x.y.t} & \text{\textcolor{gray}{weak-head neutrals}}\\
    \boxed{\nf f} & \eqdef & \can f \vee \neu f & \text{\textcolor{gray}{weak-head normal forms}} \\
    \boxed{\isTy t} & \eqdef& \univ \mid \P x : t. t \mid \nat \mid n &
      \text{\textcolor{gray}{canonical types}}
  \end{array}\]
    \vspace*{-1.5em}
  \caption{\formfile{Syntax/NormalForms.v} Canonical, neutral, and normal form predicates
  (excerpt, see \cref{app:nf}).
  Way to read the BNF-style grammars:
  \(\isTy\) is the predicate on terms given by the grammar
  \(\univ \mid \P x : t. t \mid \nat \mid n\)
  (\(n\) stands for a neutral and \(t\) for an arbitrary term).}
  \label{fig:nf}
\end{figure}


\subparagraph{Normal forms and reduction}
Amongst terms, we particularly focus on (weak-head) normal forms, 
as this is the class with good injectivity properties, defined in \cref{fig:nf}.
%
  \emph{Canonical forms}
    are terms starting with a constructor: \(\succ\), \(\l\), etc.
  \emph{Neutral forms} are ``stuck'' terms, consisting of a spine of eliminators
  on top of a variable.
  \emph{Normal forms} are either.
We also define subclasses for each type: \(\isFun t\)
holds if \(t\) is either a neutral or a \(\l\)-abstraction,
and similarly for \(\isTy\) and \(\isNat\). The \(\isPos\) predicate
isolates positive types: \(\nat\), \(\univ\) and neutrals.

\begin{figure}
  
\begin{mathpar}
  \jformlow{t \ored t'}[Term $t$ weak-head reduces in one step to term $t'$]
  \inferdef[\(\beta\)Fun]
    { }
    {(\l x : A . t)\ u \ored \sub{t}{u}}
  \quad
  \inferdef[\(\beta\)Zero]{ }{
    \natelim{0}{x.P}{b_{0}}{x.y.b_{\succ}} \ored b_{0}}
  \quad
  \inferdef[AppRed]
    {t \ored t'}
    {t~u \ored t'~u}
  \quad \dots
  \vspace*{-1em}
  %
\end{mathpar}

  \caption{\formfile{Syntax/UntypedReduction.v}  Weak-head reduction (excerpt, see \cref{app:reduction}).}
  \label{fig:reduction-main}

\end{figure}


Finally, reduction \(t \red t'\) is
the reflexive, transitive closure of one-step reduction, defined in \cref{fig:reduction-main}.
We use only a deterministic \emph{weak-head} reduction.
Reduction is \emph{not} part of the specification of typing,
we present it at this stage only to state normalisation in \cref{sec:meta-theory}.


\subsection{Meta-theory}
\label{sec:meta-theory}

Let us now \emph{state} the main meta-theoretic properties we consider ---we comment about
how to \emph{prove} them at the end of this section and in \cref{sec:related}.
In \Coq, these are represented as type-classes: most theorems take instances of the relevant
type-classes as parameters. This enables fine-grained bookkeeping of the theorems' dependency,
while keeping it lightweight as instances are automatically filled in.
This is inspired by the treatment of function extensionality and
univalence in the HoTT library~\cite{Bauer2017}.






\subparagraph{Type-level injectivity}
\begin{property}[\formline{TypingProperties/PropertiesDefinition.v}{78}{TypeInjectivityConsequencesectivity} Injectivity and no-confusion of type constructors]
  \label{prop:inj-ty}

  Assume \(T\) and \(T'\) are convertible types in weak-head normal form, \ie we have
  that \(\isTy T\), \(\isTy T'\), and \(\conv{\Gamma}{T}{T'}\).
  Then one of the following holds:
  \begin{itemize}
    \item \(T = \nat = T'\) 
      or \(T = \univ = T'\);
    \item \(T = \P x : A.B\), \(T' = \P x : A'.B'\), with \(\conv{\Gamma}{A'}{A}\) and
      \(\conv{\Gamma, x : A'}{B}{B'}\);
    \item \(T\), \(T'\) are both neutral, and \(\conv{\Gamma}{T}{T'}[\univ]\).
  \end{itemize}
  Any other case is thus impossible (\emph{no-confusion}),
  \eg we cannot have \(\conv{\Gamma}{\P x : A.B}{\nat}\).
\end{property}

This cannot be proven by mere induction, and indeed it implies a form
of non-degeneracy, namely \emph{equational consistency}:
not all types are convertible.
Consequences are numerous.

\begin{corollary}[\formline{TypingProperties/TypeInjectivityConsequences.v}{440}{subject_reduction} Preservation]
  Assuming \cref{prop:inj-ty},
  if \(\typing{\Gamma}{A}\) and \(A \red A'\) then \(\conv{\Gamma}{A}{A'}\). Similarly,
  if \(\typing{\Gamma}{t}[A]\) and \(t \red t'\) then \(\conv{\Gamma}{t}{t'}[A]\).
\end{corollary}


\begin{corollary}[\formline{TypingProperties/TypeInjectivityConsequences.v}{505}{WhClassification} Types classify normal forms]
  \label{prop:classify}
  Assuming \cref{prop:inj-ty}, consider \(t\) a normal form.
  If \(\typing{\Gamma}{t}[\P x : A.B]\), then \(\isFun t\).
  If \(\typing{\Gamma}{t}[\nat]\), then \(\isNat t\).
  If \(\typing{\Gamma}{t}[\univ]\), then \(\isTy t\).
  If \(\typing{\Gamma}{t}[T]\) and \(T\) is neutral, then \(t\) is neutral.
\end{corollary}

Preservation relies on the injectivity part, while classification uses no-confusion.
\Cref{prop:classify} is the key lemma towards progress
---the fact that every well-typed term is either a normal form or reduces---,
although we do not show this in the formalisation.
Thus, \cref{prop:inj-ty} suffices to establish safety in the
progress-and-preservation approach \cite{Wright1994}.

\subparagraph{Term-level injectivities}

Statements of injectivity (and no-confusion) for term constructors
are relatively close to those for types. 

\begin{property}[\formline{TypingProperties/PropertiesDefinition.v}{154}{nat_conv_inj} Injectivity and no-confusion at \(\nat\)]
  \label{prop:inj-nat}

  Assume \(n\) and \(n'\) are such that
  \(\isNat n\), \(\isNat n'\), and \(\conv{\Gamma}{n}{n'}[\nat]\).
  Then one of the following must hold:
  \begin{itemize}
    \item \(n = 0 = n'\);
    \item \(n = \succ(t)\), \(n' = \succ(t')\), with \(\conv{\Gamma}{t}{t'}[\nat]\);
    \item \(n\), \(n'\) are both neutral.
  \end{itemize}
\end{property}

Maybe surprisingly, the eliminator makes these statements directly provable. Indeed,
assuming \(\succ(t) \convop \succ(t')\), by the computation and congruence rule for the eliminator, we get that
\(t \convop \natelim{x.\nat}{\succ(t)}{0}{x.y.x} \convop \natelim{x.\nat}{\succ(t')}{0}{x.y.x} \convop t' \ty \nat \).
We can similarly reduce no-confusion to that for types.
A similar phenomenon happens at \(\P\) types, where injectivity of \(\l\) can also
be directly derived.

Injectivity for
the universe (\formline{TypingProperties/PropertiesDefinition.v}{149}{univ_conv_inj})
follows a similar pattern, but is not so easily proven for lack of an eliminator.
For neutral types, which have no constructors, \cref{prop:classify} is
all we need: two (convertible) normal forms at a neutral type must be neutrals.

\subparagraph{Neutral injectivities}

To complete the injectivity picture, we turn to neutrals, for which
injectivity/no-confusion can be expressed in two ways. The first
follows a familiar pattern. 

\begin{property}[\formline{TypingProperties/PropertiesDefinition.v}{224}{NeutralInj}
  Injectivity of neutral eliminators]
  \label{prop:inj-neu}
  
  Assume \(n\) and \(n'\) are neutrals such that \(\conv{\Gamma}{n}{n'}[T]\).
  Then one of the following must hold:
  \begin{itemize}
    \item \(n = x = n'\), with \((x : T') \in \Gamma\) and \(\conv{\Gamma}{T'}{T}\);
    \item \(n = m\ u\), \(n' = m'\ u'\) with \(\conv{\Gamma}{m}{m'}[\P x : A.B]\),
      \(\conv{\Gamma}{u}{u'}[A]\) and \(\conv{\Gamma}{\sub{B}{u}}{T}\);
    \item \(n = \natelim{m}{x.P}{t_0}{x.y.t_{\succ}}\),
      \(n' = \natelim{m'}{x.P'}{t'_0}{x.y.t'_{\succ}}\),
      with convertible subterms.
  \end{itemize}
\end{property}

However, in some cases this is too strong: at negative types our type-directed algorithm eagerly
\(\eta\)-expands, so the above are only
really needed at positive types. To state the ``minimal'' property we require,
we take another way, starting with the following definition.

\begin{definition}[\formline{DeclarativeTyping.v}{336}{DeclNeutralConversion} Neutral comparison]
  \label{def:neutral-comp}
  \emph{Neutral comparison} (see \cref{app:neutral-comp}),
  written \(\nconv{\Gamma}{n}{n'}{T}\), is the least relation that is:
  \begin{itemize}
    \item reflexive on variables, \ie \(\nconv{\Gamma}{x}{x}{T}\) if \((x : T) \in \Gamma\);
    \item closed under the congruence rules for eliminators;
    \item stable under conversion: if \(\nconv{\Gamma}{n}{n'}{T}\) and
      \(\conv{\Gamma}{T}{T'}\) then \(\nconv{\Gamma}{n}{n'}{T'}\).
  \end{itemize}
\end{definition}

We can then rephrase injectivity of neutral constructors as follows.

\begin{property}[\formline{TypingProperties/PropertiesDefinition.v}{185}{ConvNeutralConv}
  Completeness of neutral comparison]
  \label{prop:neu-compl}
  Given any two neutrals \(n\) and \(n'\) such that \(\conv{\Gamma}{n}{n'}[T]\),
  we have \(\nconv{\Gamma}{n}{n'}{T}\).
\end{property}

\begin{proposition}[%
  \formline{TypingProperties/NeutralConvProperties.v}{189}{neu_inj_conv_neu}]

  \Cref{prop:inj-neu} and \cref{prop:neu-compl} are equivalent.

\end{proposition}


But now, we can weaken it to only consider \emph{positive types}.

\begin{property}[\formline{TypingProperties/PropertiesDefinition.v}{177}{ConvNeutralConvPos}
  Completeness of neutral comparison at positive types]
  \label{prop:neu-compl-pos}
  Given any two neutrals \(n\) and \(n'\) such that \(\conv{\Gamma}{n}{n'}[T]\),
  if moreover \(\isPos T\), then we have \(\nconv{\Gamma}{n}{n'}{T}\).
\end{property}

\subparagraph{Normalisation} The last property is
normalisation. 
Its main use is to argue for termination of the algorithm, so it
has to go beyond weak-head reduction and take into account
\(\eta\)-expansion, and subterms of (weak-head) normal forms.

\begin{definition}[\formline{TypingProperties/NormalisationDefinition.v}{14}{dnorm_tm}
  Deeply normalising]
  \label{def:normalisation}
  A term \(t\) is \emph{deeply normalising} at type \(A\) in context \(\Gamma\)
  if they reduce respectively to weak-head normal forms \(t'\) and \(A'\),
  and moreover:
  \begin{itemize}
    \item \(A'\) is a negative type, and the \(\eta\)-expansion of \(t'\)
      is itself normalising;
    \item or \(A'\) is a positive type, and the subterms of \(t'\) are
      normalising at the relevant types.
  \end{itemize}
  A type \(A\) is \emph{deeply normalising} in context \(\Gamma\)
  if it reduces to a weak-head normal form, whose subterms are themselves normalising.
  See \cref{sec:deep-norm} for the inference rules.
\end{definition}

While not apparent in this short version of the definition, the context actually
plays a role in order to express the type at which a neutral should be
normalising.


\begin{property}[\formline{TypingProperties/PropertiesDefinition.v}{240}{DeepNormalisation}
  Deep normalisation]
  \label{prop:deep-norm}

  Every well-typed term is deeply normalising (at its type). Every well-formed type
  is deeply normalising.
  
\end{property}

Meta-theoretic consequences include canonicity
(\formline{TypingProperties/NormalisationConsequences.v}{307}{nat_canonicity})
---any closed
term of type \(\nat\) is convertible to one of the form \(\succ^{n} 0\)---
and, most importantly, logical consistency. This means that to prove
\cref{prop:deep-norm}, the meta-theory must be logically stronger than the object
theory.

\begin{corollary}[%
  \formline{TypingProperties/NormalisationConsequences.v}{282}{consistency}
  Logical consistency]
  \label{prop:consistency}
  Assuming \cref{prop:deep-norm}, there are no closed terms of type \(\Empty\).
\end{corollary}

\subparagraph{\formfile{TypingProperties/LogRelConsequences.v} Proving the properties}
In our formalisation, most properties
are proven via a logical relation, in the style pioneered
by Abel \etal\ \cite{Abel2017}, and later adapted in \Coq \cite{Adjedj2024},
to which we refer for details. Multiple universes in the meta-theory
provide the logical strength we need to prove consistency.
This approach is however by far not the only possibility,
see the discussion in \cref{sec:related}. Note the following two points.

First, the proof of deep normalisation is indirect. Indeed, the logical
relation is parameterised by an abstract partial equivalence relation (PER) relating neutrals,
which we instantiate with algorithmic neutral comparison
to obtain normalisation for subterms of neutrals.
To avoid this detour through algorithmic aspects,
we envisioned strengthening the logical relation to directly
encompass neutrals, but did not find a satisfying approach.

Second, as logical relations at negative types
are naturally defined by observations
(\eg a function is reducible when its application to reducible inputs are
reducible), these types ``know nothing'' about neutrals, and
completeness of neutral comparison at negative types cannot be derived easily.
We expect other proofs of ``semantic'' flavour
to have 
similar issues with full completeness of neutrals,
making the restriction to positive types (\cref{prop:neu-compl-pos} vs.\ \cref{prop:neu-compl})
an important distinction.
And indeed, while completeness at positive types is a direct
consequence of our logical relation (\formline{TypingProperties/LogRelConsequences.v}{380}{ConvNeutralConvPosLogRel}), the approach we found to completeness at all types takes
a long detour via the two algorithmic presentations
(\formline{Algorithmic/UntypedTypedConv.v}{684}{conv_pos_all})
and strengthening.

\section{Algorithmic typing and conversion}
\label{sec:algorithms}

Let us now turn to our algorithms for conversion and typing.
Only the untyped conversion algorithm is new, the typed conversion and
bidirectional typing ones come directly from MLTT \textit{à la} \textsc{Coq} \cite{Adjedj2024}.
We give the algorithms informally as rules, which are translated
in the formalisation in two ways. First (\formfile{AlgorithmicJudgments.v})
as inductively defined relations, like declarative typing.
%
Second, we implement them as functions (\formfile{Checkers/Functions.v}).
As we wish to separate their definition from their termination
argument, we must embrace partiality.
To that end, we use McBride's free recursion monad \cite{McBride2015},
via Winterhalter's \textsc{PartialFun} \cite{Winterhalter2023}.
This lets us define the functions prior to any proofs;
immediately execute them with bounded recursion depth
---McBride's ``petrol semantics''; extract a typing (dis)proof from a terminating
execution once we prove either soundness; and get proper decidability
once we show normalisation. 
Cherry on the cake: as the recursion monad exposes
the programs' recursive structure, we can generically derive
``functional'' induction principles, following said recursive structures.

\subparagraph{\formline{AlgorithmicJudgments.v}{153}{WfTypeAlg} Bidirectional typing}

Our typing algorithm, as most for dependent types, is bidirectional:
it separates \emph{inference}, written \(\inferty{\Gamma}{t}{T}\),
where the type is an output to be found, from type \emph{checking},
written \(\checkty{\Gamma}{t}[T]\), where the type is a constraint given as input.
We refer to others \cite{McBride2018,Dunfield2021,LennonBertrand2022PhD}
for exposition ---this last reference is the closest
to us. Bidirectional typing generally has two main interests,
although only the first is relevant to us here.

First, bidirectional typing gives better control over conversion. By replacing
the conversion rule with syntax-directed rules,
we get a term-directed system,
amenable to implementation. 
This makes the bidirectional system \textit{a priori} more restrictive than the
declarative one, since we have strictly constrained the way conversion can be used.
The main point of negative soundness is to show that this restriction does not,
in fact, lose anything.

Second, propagating types ``up'' in derivations
lightens the need for type annotations and gives better error messages.
This is important for surface languages, less for kernels,
where terms are typically not human-written.
On the contrary, extra information can be useful in various ways,
and we would rather dispense of the complication of
segregating terms containing too little information to infer. Hence, as
\Coq and \Lean, we adopt a syntax where all terms infer, although
completeness for the other approach is certainly interesting and
feasible~\cite{LennonBertrand2023}.

\begin{figure}
  \begin{mathpar}
    \jform{\bconv*{\Gamma}{t}{t'}[A]}[
      Reduced terms \(t\) and \(t'\) are convertible at type \(A\)]
   
    \inferdef[TFun]
    {\bconv{\Gamma}{A}{A'}[\univ] \\
      \bconv{\Gamma, x : A'}{B}{B'}[\univ]}
    {\bconv*{\Gamma}{\P x : A . B}{\P x : A'. B'}[\univ]}
    \and
    \inferdef[FunExp]
    {\bconv{\Gamma, x : A}{f~x}{f'~x}[B]}
    {\bconv*{\Gamma}{f}{f'}[\P x : A.B]}
    \label{rule:fun-exp}
    \and
    \inferdef[TSucc]
      { \bconv{\Gamma}{t}{t'}{\nat}}
      {\bconv*{\Gamma}{\succ(t)}{\succ(t')}{\nat}} \and
    \inferdef[NePos]
    {\isPos T \\ \bnconv{\Gamma}{n}{n'}{S}}
    {\bconv*{\Gamma}{n}{n'}[T]} \and
    \label{rule:neu-pos} \dots \and
    \jformlow{\bconv{\Gamma}{t}{t'}[A]}[Terms \(t\) and \(t'\) are convertible at type \(T\)]
    \jformlow{\bnconv*{\Gamma}{t}{t'}{T}}[
      Neutrals \(t\) and \(t'\) are comparable, inferring the reduced type \(T\)]
    \inferdef[TTmRed]
    {T \red U \\ t \red u \\ t' \red u' \\\\
      \bconv*{\Gamma}{u}{u'}[U]}
    {\bconv{\Gamma}{t}{t'}[T]} \and
    \inferdef[NRed]
      {\bnconv{\Gamma}{n}{n'}{T} \\ T \red S}
      {\bnconv*{\Gamma}{n}{n'}{S}} \\
    \jform{\bnconv{\Gamma}{t}{t'}{T}}[
      Neutrals \(t\) and \(t'\) are comparable, inferring the type \(T\)]
    \inferdef[TVar]
      {(x : T) \in \Gamma}
      {\bnconv{\Gamma}{x}{x}{T}}
    \label{rule:tconv-var}
    \and
    \inferdef[TApp]
      {\bnconv*{\Gamma}{n}{n'}{\P x : A . B} \\ \bconv{\Gamma}{u}{u'}[A]}
      {\bnconv{\Gamma}{n~u}{n'~u'}{\sub{B}{u}}}
      \label{rule:napp}
    \and \dots
  \end{mathpar}

  \caption{\formline{AlgorithmicJudgments.v}{14}{ConvTypeAlg} Typed conversion algorithm
    (excerpt, see \cref{app:algo-typed-conv}).}
  \label{fig:typed-algo}
\end{figure}

\subparagraph{\formline{AlgorithmicJudgments.v}{14}{ConvTypeAlg} Typed conversion is bidirectional}

A point already implicitly present in Abel \etal~\cite{Abel2017}, but
made explicit in Adjedj \etal~\cite{Adjedj2024}, is that the
type-directed conversion algorithm is bidirectional, too. It
is decomposed in two main functions, one to compare general terms and the other
dedicated to neutrals. The former takes a type as input ---it is
\emph{checking}---, while the latter reconstructs types as it traverses
the neutrals ---it is \emph{inferring}.

More precisely, conversion (see \cref{fig:typed-algo}), written
\(\bconv{\Gamma}{t}{u}[A]\), goes as follows:
\begin{enumerate}
  \item reduce \(t\), \(u\) and \(A\) to weak-head normal forms;
  \item if \(A\) is a negative type (\(\P,\Sig\)),
    recursively compare the \(\eta\)-expansions of \(t\) and \(u\);
  \item if \(A\) is a positive type (\(\nat, \univ, \dots\))
   and the terms are canonical, use the relevant congruence;
  \item if \(A\) is a positive type and the terms are neutrals, call neutral comparison.
\end{enumerate}
Steps 2-4 are delegated to an auxiliary function, denoted \(\bconvop{\BooleanTrue}\)
(the ``h'' is for ``head'').
Neutral comparison, written \(\bnconv{\Gamma}{n}{n'}{S}\)
follows their structure, applying congruences as it goes.

\ruleref{rule:napp} demonstrates this bidirectional yoga.
As conversion is type-directed, we need a type \(A\) to compare the arguments,
that is obtained through the recursive comparison
of the functions, which must thus be inferring.
We would not be able to construct \(\P x : A.B\) anyway,
even given \(\sub{B}{u}\): inverting substitutions is unfeasible.
Thus, even with a syntax with complete inference, the distinction between
checking constructors and inferring eliminators, a landmark of
the ``Nordic'' approach to bidirectionalism \cite{Norell2007,McBride2018},
appears naturally.

\subparagraph{\formline{AlgorithmicJudgments.v}{278}{UConvAlg} Untyped conversion}

\begin{figure}
\begin{mathpar}
  \jformlow{\uconv{t}{t'}}[Terms \(t\) and \(t'\) are convertible]
  \inferdef[UTmRed]
  {t \red u \\ t' \red u' \\
    \uconv*{u}{u'}}
  {\uconv{t}{t'}}
  
  \jformlow{\uconv*{t}{t'}}[
    Reduced terms \(t\) and \(t'\) are convertible]
  
  \inferdef[UUni]
    { }
    {\uconv*{\univ}{\univ}}
  \and
  \inferdef[UFun]
    {\uconv{A}{A'} \\ \uconv{B}{B'}}
    {\uconv*{\P x : A . B}{\P x : A'. B'}}
  \and
  \inferdef[ULamLam]
    {\uconv{t}{t'}}
    {\uconv*{\l x : A. t}{\l x : A'. t'}}
  \label{rule:un-lam}
  \and
  \inferdef[ULamNe]
    {\neu n' \\ \uconv{t}{n'\ x}}
    {\uconv*{\l x : A. t}{n'}}
  \label{rule:un-lam-ne}
  \and
  \inferdef[USucc]
    {\uconv{t}{t'}}
    {\uconv*{\succ(t)}{\succ(t')}} \and
  \inferdef[UNeNe]
    {\unconv{n}{n'}}
    {\uconv*{n}{n'}}
    \label{rule:conv-ne-ne}
    \and \dots \and
  \jformlow{\unconv{n}{n'}}[
    Neutrals \(n\) and \(n'\) are comparable]
  \inferdef[UVar]
    { }
    {\unconv{x}{x}}
  \label{rule:uconv-var}
  \and
  \inferdef[UApp]
    {\unconv{n}{n'} \\ \uconv{u}{u'}}
    {\unconv{n\ u}{n'\ u'}}
  \and \dots
\end{mathpar}

\caption{\formline{AlgorithmicJudgments.v}{278}{UConvAlg}
  Untyped algorithmic conversion (excerpt, see \cref{app:algo-untyped-conv}).}
\label{fig:untyped-algo}
\end{figure}

Untyped conversion (\cref{fig:untyped-algo})
looks much simpler than typed conversion: there
are no types, no separation between type and term conversion, and the separation
of neutral comparison is kept mostly for readability.

The main point is that we replace the all-encompassing \nameref{rule:fun-exp}
by term-directed rules, an approach pioneered by Coquand,
first for functions \cite{Coquand1996}, then pairs~\cite{Abel2007}.
When comparing \(\l\)-abstractions, we apply the congruence \nameref{rule:un-lam}
---which would also happen in \nameref{rule:fun-exp} after contracting the redexes.
If one side is a \(\l\) and the other a neutral
(\nameref{rule:un-lam-ne}), we again simulate \nameref{rule:fun-exp},
but this time we see the application to a fresh variable is visible on the neutral's side.
And this is all! That is, \emph{there is no \(\eta\)-rule when the two sides are neutrals}.
How could there be? In the extreme case of two variables, the only
information we could use is... type information, which we precisely decided not to rely on.
This discrepancy on neutrals at negative types is the key difference
between the two algorithms. 

\subparagraph{Comparison}
As a concrete example, we can look at the ``execution trace'' in the following problem:
\(\conv{x : \nat \to \nat}{x}{x}[\nat \to \nat]\).
The untyped algorithm immediately uses \nameref{rule:conv-ne-ne} to go in ``neutral
mode'', and finishes with \nameref{rule:uconv-var}. The typed algorithm, on the other hand,
first \(\eta\)-expands, introducing a fresh variable \(y\),
and recursively compares \(x\ y\) to \(x\ y\). Only then, by reaching a positive type,
does it change to ``neutral mode''. Neutral comparison then
peels off the application node, compares \(x\) to \(x\) again but this time as neutrals,
and finally applies variable reflexivity.
It also needs to compare \(y\) to \(y\) first as normal forms at \(\nat\),
then as neutrals, and finally terminates.
We can already see the work spared by the term-directed algorithm
---exacerbated for \(\Sig\) types, where \(\eta\) introduces
duplication.


\section{Properties of the algorithms}
\label{sec:certification}

\begin{figure}
  \begingroup
  \setlength{\fboxsep}{0pt}%
  \renewcommand{\arraystretch}{1}
  \setlength{\tabcolsep}{8pt}
  \newcommand{\widewedge}{\: \wedge \:}
  \newcommand{\implyq}{\ensuremath{\stackrel{?}{\Rightarrow}}}

  \rowcolors{2}{gray!30}{white}
  \begin{tabular}{ccccc}
    Raw judgment && Precondition && Postcondition \\
    \hline
    \(\inferty{\Gamma}{t}{T}\) & \(\wedge\) & $\ctxtyping{\Gamma}$ &\implyq &
      \(\typing{\Gamma}{t}[T] \widewedge T \text{ unique}\) \\
    $\checkty{\Gamma}{t}{T}$ & \(\wedge\)
      & $\typing{\Gamma}{T}$ &\implyq&
      \(\typing{\Gamma}{t}[T]\)\\
    $\bconv{\Gamma}{T}{T'}$ &\(\wedge\)&
      $\typing{\Gamma}{T} \widewedge \typing{\Gamma}{T'}$
      & \implyq & \(\conv{\Gamma}{T}{T'}\) \\
    $\bconv{\Gamma}{t}{t'}[T]$ &\(\wedge\)&
      $\typing{\Gamma}{t}[T] \widewedge \typing{\Gamma}{t'}[T]$
      & \implyq & \(\conv{\Gamma}{t}{t}[T]\) \\
    $\bnconv{\Gamma}{t}{t'}{T}$ &\(\wedge\)&
      \makecell{
        \(\neu{t} \widewedge \exists A.\ \typing{\Gamma}{t}[A] \widewedge \) \\
        \(\neu{t'} \widewedge \exists A'.\ \typing{\Gamma}{t'}[A']\)
      }
      & \implyq &
        \makecell{\(\conv{\Gamma}{t}{t'}[T] \widewedge\) \\ \(T \text{ unique}\)} \\
    $\uconv{T}{T'}$ &\(\wedge\)&
      $\exists \Gamma.\ \typing{\Gamma}{T} \widewedge \typing{\Gamma}{T'}$
      & \implyq & \(\conv{\Gamma}{T}{T'}\) \\
    $\uconv{t}{t'}$ &\(\wedge\)&
      $\exists \Gamma\ T.\ \typing{\Gamma}{t}[T] \widewedge \typing{\Gamma}{t'}[T]$
      & \implyq & \(\conv{\Gamma}{t}{t}[T]\) \\
    $\unconv{t}{t'}$ &\(\wedge\)&
      \colorbox{gray!30}{\makecell{
        \(\neu{t} \widewedge \neu{t'} \widewedge\) \\
        \(\exists \Gamma\ A\ A'.\ \typing{\Gamma}{t}[A] \widewedge
        \typing{\Gamma}{t'}[A']\)
      }}
      & \implyq &
      \colorbox{gray!30}{\makecell{
        \(\conv{\Gamma}{t}{t'}[T] \widewedge\) \\ \(T \text{ unique}\)}} \\
  \end{tabular}
\endgroup
  
\caption{\formfile{Algorithmic/Bundled.v}
  Pre- and post-conditions of the algorithmic judgments.
  }
\label{fig:prepost}

\end{figure}

An important piece is
missing from the algorithmic judgments as defined so far: invariants. Indeed,
there is a number of things that are not directly checked.
For instance, \nameref{rule:tconv-var} does not enforce context well-formation,
and \nameref{rule:neu-pos} does not check that the inferred type
coincides with the input one. These are not unfortunate oversight;
instead, we avoid re-doing work that can be obtained for free through invariants,
which is crucial in implementations ---constantly re-checking context well-formation,
for instance, is terrible for performance.

Hence, each function has preconditions, which should
be true before we even dare ask the question that the function ought to answer;
and a postcondition, which is what it is in charge of establishing.
These are summarised in \cref{fig:prepost}. By ``\(T\) unique''
we mean that any other type also satisfying the same property must be convertible
to \(T\).
``Reduced'' judgments, \eg \(\bconv*{\Gamma}{T}{T'}\),
have the same postconditions as their unreduced counterparts,
and the extra precondition that inputs should be weak-head normal forms.

As we can see, although no type is visible in the definition of untyped conversion,
the name ``untyped'' is misleading. We still aim to decide the same typed relation,
but have merely remarked that we can arrange the algorithm in a way that makes type information
superfluous. Yet types are still present in invariants,
silently keeping the algorithm on rails.
Untyped conversion may thus be more aptly characterised as
\emph{term-directed typed conversion}.

\subsection{Properties of typed conversion}

Most results in this section are not \textit{per se} new, although they were significantly
reworked, in particular to precisely track their dependencies. In particular,
termination relies explicitly on normalisation, rather than on the reflexivity
of the logical relation. The work around negative soundness is entirely new.

\subparagraph{Positive soundness}

We split the proof in two, using the inductive judgment
as intermediate. The first half is almost tautological, and does
not depend on any meta-theory.

\begin{proposition}[\formline{Checkers/Soundness.v}{127}{implem_tconv_sound}
  Positive soundness of the typed conversion function]
  \label{prop:sound-typed-fun}
  If a typed conversion function returns a positive answer,
  then the corresponding judgment of \cref{fig:typed-algo} holds.
\end{proposition}

\begin{proposition}[\formline{Algorithmic/Bundled.v}{1198}{algo_conv_sound}
  Positive soundness of the typed algorithmic conversion judgment]
  \label{prop:sound-typed-judg}
  Assuming injectivity of type constructors (\cref{prop:inj-ty}), 
  the algorithmic typed conversion judgments
  imply their postconditions whenever their preconditions hold.
\end{proposition}

This second half is the interesting one. The crucial part of this proof,
and the backbone of others in this section, is to show
that each inductive rule correctly preserves the invariants.
This is a very regular property, so much that we meta-programmed%
\footnote{
  We hardcoded these in the current code to remove a dependency
  which caused proof engineering issues.
}
the generation of ``local preservation lemmas''
(\formline{Algorithmic/Bundled.v}{290}{Invariants})
for each rule, expressing this idea.
%
We assemble these lemmas into a
``bundled induction principle'', again programmatically generated,
which threads the invariants through,
giving access to them as extra hypotheses in each induction step.
Soundness follows by instantiating this induction principle with the
constant \texttt{True} predicates.

\subparagraph{Negative soundness}

For negative soundness, we cannot go through the intermediate algorithmic judgments,
as they only encode successful runs of the functions. Thus, we bite the bullet and
directly relate the output of the functions to the declarative judgments.

\begin{proposition}[\formline{Checkers/NegativeSoundness.v}{426}{implem_tconv_sound_neg}
  Negative soundness of typed algorithmic conversion]
  \label{prop:neg-sound-typed}
  Assuming injectivity of type and term constructors
  (Properties \ref{prop:inj-ty}, \ref{prop:inj-nat}, etc.),
  and completeness of neutral conversion \emph{at positive types} (\cref{prop:neu-compl-pos}),
  if a typed conversion function returns a negative answer and its precondition
  holds, then its postcondition cannot hold.
\end{proposition}

Again, we rely on invariant preservation, the proof is otherwise straightforward: we have
extracted exactly the injectivity properties needed to justify each step.
In essence, we show that at any stage the current ``remaining goals''
are equivalent to the initial query,
relying on injectivity to know that we only use invertible congruences,
which are safe to apply.

\subparagraph{Termination}

As expected, this is where normalisation becomes necessary.
As before, we need invariants throughout, and thus injectivity of type constructors.
The most subtle part of the proof concerns the reduction machine.

\begin{lemma}[%
  \formline{Checkers/ReductionCorrectness.v}{363}{wh_red_complete_whnf_tm}
  Completeness of reduction]
  Assuming \cref{prop:inj-ty}, if
  \(t\) is well-typed and reduces to the weak-head normal form \(t'\),
  then the reduction function applied to \(t\) returns \(t'\).
\end{lemma}

The main task is to define a well-founded order on the configurations
of the abstract machine implementing reduction,
a lexicographic product of subterm ordering and reduction.
The assumption that \(t\) reduces to a weak-head normal
form ensures that this order is well-founded. Winterhalter introduced a
similar order for \MetaCoq \cite[chap. 23]{Winterhalter2020}.

\begin{proposition}[%
  \formline{Checkers/Termination.v}{455}{tconv_terminates}
  Termination of typed algorithmic conversion]
  Assuming injectivity of type constructors (\cref{prop:inj-ty})
  and deep normalisation (\cref{prop:deep-norm}),
  typed conversion functions terminate whenever their preconditions are
  satisfied.
\end{proposition}

The proof is by induction on deep normalisation, again using local preservation lemmas.
The inductive structure of deep normalisation readily encodes the combination of subterm,
reduction and \(\eta\)-expansion that one needs to be well-founded.
In contrast, in \MetaCoq, Winterhalter introduced
a complex dependent lexicographic order modulo \cite[chap 24]{Winterhalter2020}, 
and admits that his setup would not easily incorporate \(\eta\)-expansion.

\subparagraph{Non-triviality}
Soundness, both positive and negative, has a default: the nowhere-defined function
is sound! Although we repeat that these properties should not be
read to merely say ``there exists a function such that...''
but to characterise a particular checker ---that we can run on examples
(\formfile{Checkers/Execution.v}) to check for non-triviality---,
it is good to formally ensure that our functions are non-trivial. 
Again, this is split in two.

\begin{proposition}[%
  \formline{Checkers/Completeness.v}{254}{implem_conv_complete_ty}
  Completeness of the implementation]
  Assuming injectivity of type constructors (\cref{prop:inj-ty}),
  if an algorithmic conversion judgment and its precondition hold,
  then the corresponding function answers positively.
\end{proposition}

Next, we show that algorithmic conversion satisfies many closure properties,
which Abel \etal\ call ``generic conversion''.
These are rather rich: all the proofs on the logical
relation side can be done with respect to \emph{any} generic conversion.
\begin{proposition}[%
  \formline{Algorithmic/TypedConvInstances.v}{376}{IntermediateTypingProperties}]
  Algorithmic conversion is a generic conversion, thus for instance symmetric
  and transitive, stable by weakening and
  expansion, congruent for constructors.
\end{proposition}
Compared to declarative typing, the only missing part is congruences for
eliminators, which only hold for neutral conversion,
and reflexivity ---which would be direct by induction if we had all congruence rules.
The point is that eliminator congruences are the dangerous ones
that can trigger reductions: applying a normal form to a normal form does not yield
a normal form.
And indeed a term is deeply normalising exactly when it is reflexive for
(typed) algorithmic conversion, explaining why reflexivity cannot be achieved easily.

\subparagraph{Typing}

All of these propositions extend to bidirectional typing. More precisely, assuming
\emph{any} conversion relation/function between types that satisfies the
relevant property, we can derive a similar one for typing.
This means that the proofs are modular, in that we
can swap out a different implementation of conversion, as long as it satisfies
the right properties, which is exactly what we are about to do.

\subsection{Properties of untyped conversion}
\label{sec:untyped-properties}

The high-level ideas are very similar to typed conversion.
For positive soundness there is really no difference:
the equivalents of Propositions \ref{prop:sound-typed-fun}
(\formline{Checkers/Soundness.v}{197}{implem_uconv_sound})
and \ref{prop:sound-typed-judg}
(\formline{Algorithmic/UntypedConvSoundness.v}{493}{uconv_sound_decl})
hold. Even better, we can reuse most of the local preservation lemmas directly.

Negative soundness is where the subtle difference between the two algorithms shows.
Echoing the discussion in \cref{sec:meta-theory,sec:algorithms}, the key point
is neutrals at negative types. And indeed,
while in \cref{prop:neg-sound-typed} we could make do with the weaker
completeness (\cref{prop:neu-compl-pos}), this time we really need the full power
of completeness at all types.

\begin{proposition}[%
  \formline{Checkers/UntypedNegativeSoundness.v}{652}{implem_uconv_sound_neg}
  Negative soundness of untyped algorithmic conversion]
  Assuming injectivity of type and term constructors
  (Properties \ref{prop:inj-ty}, \ref{prop:inj-nat}, etc.),
  and completeness of neutral conversion \emph{at all types} (\cref{prop:neu-compl}),
  if the untyped conversion function returns a negative answer and its preconditions
  holds, then its postcondition cannot hold.
\end{proposition}

As for termination, there is a twist again, of course at negative types. The issue
is that deep normalisation systematically \(\eta\)-expands terms, but untyped conversion
does not always do this. Worse, it is non-deterministic, in the sense that whether
a given term is \(\eta\)-expanded depends on the term it is compared to. Thus, our
termination measure is not exactly fit, and circumventing the issue involves
proving a form of strengthening of the algorithm: it gives the same output when
irrelevant variables are removed from the context.

\begin{proposition}[%
  \formline{Checkers/UntypedTermination.v}{851}{uconv_terminates}
  Termination of untyped algorithmic conversion]
  Assuming \cref{prop:inj-ty} and \cref{prop:deep-norm},
  untyped conversion functions terminate whenever their preconditions hold.
\end{proposition}

\subsection{Equivalence between the algorithms}
\label{sec:algo-equiv}

To wrap up, we can directly compare the two algorithms.
Given how close they are, this should be rather easy, or at any rate
easier than comparing declarative systems ---where the seemingly innocuous
transitivity is actually a source of headaches \cite{Siles2012},
because it means ``real'' untyped conversion can relate well-typed terms through
ill-typed ones, which cannot be easily emulated by its typed counterpart.
In fact, streamlining the relation between typed and untyped conversion
was one of our original motivation
\cite[chap. 6]{LennonBertrand2022PhD} \cite{LennonBertrand2022b}.

\begin{proposition}[%
  \formline{Algorithmic/UntypedTypedConv.v}{39}{bundled_conv_uconv}
  Typed to untyped conversion]
  \label{prop:typed-to-untyped}
  Assuming \cref{prop:inj-ty}, the typed algorithmic judgments
  imply their untyped counterparts whenever their preconditions hold.
\end{proposition}

\begin{proof}
The proof is by bundled induction on the typed algorithmic judgment. We simultaneously prove 
that if \(\bconv{\Gamma}{n}{n'}[T]\) and moreover \(n\), \(n'\) are neutrals,
then \(\unconv{n}{n'}\). 

The main interesting case is of course that of \nameref{rule:fun-exp}.
In this case, the typing invariants imply, by \cref{prop:classify}, that
both sides are either abstractions or neutrals. In the first three cases, we can directly
conclude by induction hypothesis, up to \(\beta\)-reduction of the \(\eta\)-expanded
abstraction. Remains the case of two neutrals \(n\) and \(n'\). The extra
induction hypothesis for neutrals tells us that \(\unconv{n\ y}{n'\ y}\), or more precisely
in de Bruijn syntax, that \(\unconv{\sub{n}{\uparrow}\ y}{\sub{n'}{\uparrow}\ y}\),
where \(\uparrow\) is the substitution shifting all indices by 1.
We can easily invert this to \(\unconv{\sub{n}{\uparrow}}{\sub{n'}{\uparrow}}\),
and conclude by an auxiliary lemma that untyped conversion admits strengthening.
\end{proof}
Note the final use of strengthening: while this is readily proven by
induction for algorithmic judgments, it is generally
difficult to prove for declarative systems! Admitting easy proofs of strengthening
is a major advantage of algorithmic presentations \cite{LennonBertrand2021}.

For the converse, there is a catch: normalisation is required.
Indeed, imagine a type \(A\), well-formed
in \(\Gamma\), such that there is an infinite reduction sequence \(A \ored A_1 \ored \dots\)
Trying to deduce \(\conv{\Gamma,x : A}{x}{x}[A]\), the untyped algorithm
would terminate immediately, while the typed one would diverge on the evaluation of \(A\).
More insidiously, if a type \(B\) keeps on exposing \(\P\) constructors
---say a solution to \(B \convop \nat \to B\)--- comparing \(x : B\) to
itself would terminate immediately on one side,
and spin in never ending \(\eta\)-expansions
on the other. We can, however, prove the following crucial lemma.

\begin{lemma}[\formline{Algorithmic/UntypedTypedConv.v}{148}{ne_conv_conv}
  Completeness of typed neutral conversion at all types]
  \label{lem:neu-compl-algo}
  Assuming \cref{prop:deep-norm},
  given a context \(\Gamma\) and \(A\), \(m\), \(n\) respectively a well-formed
  type and well-typed terms in \(\Gamma\), such that \(\bnconv{\Gamma}{m}{n}{A'}\)
  and \(\conv{\Gamma}{A'}{A}\), we have \(\bconv{\Gamma}{m}{n}{A}\).
\end{lemma}

\begin{proof}[Proof sketch\footnotemark]
  By induction on the deep normalisation of \(A\), using which,
  after a series of \(\eta\)-expansions of \(m\) and \(n\),
  we end up at a positive type. At this point we go into neutral mode, peel off
  all introduced \(\eta\)-expansions, and use our neutral comparison hypothesis,
  adequately weakened to account for the extra introduced variables, to conclude.
  %
\end{proof}
\footnotetext{In the formalisation, we shortcut this proof using
  the equivalence of declarative and algorithmic typing.}

\begin{proposition}[\formline{Algorithmic/UntypedTypedConv.v}{211}{uconv_tconv}
  Untyped to typed conversion]
  Assuming injectivity of typed constructors and deep normalisation,
  the untyped algorithmic judgments
  imply their typed counterparts whenever their preconditions hold.
\end{proposition}

After \cref{lem:neu-compl-algo}, remains a straightforward induction
---as usual, using preservation lemmas as necessary.
However,  let us emphasise that while the proof of equivalence is
relatively direct, it nonetheless relies on difficult properties:
strengthening, and normalisation.

\subsection{Extensions and limits}
\label{sec:discussion}

While our language is simplified compared to real proof assistants,
we believe that the features we include are enough to discover important subtleties,
and that our general methodology would scale to the type theories of \Coq, \Lean or
\Agda.%
\footnote{Barring complex extensions such as \textsc{Cubical Agda}, which would warrant further
investigations.}
The main aspect that we did not cover is definitional uniqueness, \ie
types with one inhabitant up to conversion.
Still, in the light of what precedes, we can comment on it.

\subparagraph{Definitional unit} The \(\eta\) rule for the
unit type says that any two inhabitant are convertible:
\begin{mathpar}
  \inferdef[Unit]{\typing{\Gamma}{t}[\unit] \\ \typing{\Gamma}{t'}[\unit]}
    {\conv{\Gamma}{t}{t'}[\unit]}
  \label{rule:unit}
\end{mathpar}
In a sense, this is the ultimate type-directed rule. Accordingly, it 
completely wrecks completeness of neutral comparison, since \emph{any} two neutrals
are convertible. This also ``leaks'' to other negative types via their own
\(\eta\)-laws. For instance, \(\nat \to (\Sig x : \unit . \unit)\) is
also ``unit-like'', all its elements being convertible,
and neutral comparison is likewise incomplete.

Different strategies are adopted to circumvent this issue in major proof assistants, where
\(\unit\) typically appears as a record type with no fields.
\Agda is the least affected, as its conversion is type-directed, although the subtle
behaviour of unit-like types is a regular source of bugs \cite{EtaUnitAgda4,EtaUnitAgda3,EtaUnitAgda1,EtaUnitAgda2}. \Coq instead
commits to untyped conversion, but must therefore enforce all record types
to have at least one relevant field, forbidding the definitional unit type.
\Lean takes an intermediate solution, re-inferring
the type of neutrals on the fly if their untyped comparison fails, and implements a dedicated
criterion to detect which types are ``unit-like'',%
\footnote{Although this criterion is currently incomplete, causing issues in the type-checker
  \cite{EtaUnitLean}.}
a strategy also explored by Kovács in normalisation by evaluation \cite{Kovacs2025}.

\subparagraph{Strict propositions}
The universe of strict propositions \cite{Gilbert2019}---called \texttt{Prop}
in \Agda and \Lean, and \texttt{SProp} in \Coq, which we write \(\SProp\)---,
reflects the idea of propositions being proof irrelevant: any two proofs
are definitionally equal. Concretely, the rule is as follows:
\begin{mathpar}
  \inferdef[SProp]{\typing{\Gamma}{t}[P] \\ \typing{\Gamma}{t'}[P] \\ \typing{\Gamma}{P}[\SProp]}
    {\conv{\Gamma}{t}{t'}[P]}
  \label{rule:sprop}
\end{mathpar}
This is very similar to definitional unit,
and of course comes with similar threat to completeness of neutral comparison.
\Agda's conversion, being type-directed---and already maintaining universe information---,
was easily extended. In \Lean, the same strategy
as for unit ---re-inferring types in case of failure--- is used. In combination with heavy
sharing and little type-level computation, this seems enough to keep decent performances.

\Coq, however, does something special to remain purely term-directed \cite{Gilbert2020,Leray2022}.
The core remark is that there is an important difference between \nameref{rule:sprop}
and \nameref{rule:unit}:
computing types is in general costly ---as one needs to compute substitutions,
evaluate terms, etc.--- but merely maintaining information about the \emph{universe} is much cheaper.
Moreover, since strict propositions form a universe of definitionally irrelevant types,
by construction closed under as many operations as possible,
the proliferation is contained to \SProp{} only.
For instance any function type with codomain in \SProp{} is again in \SProp{}.
Thus, universe information can be used to implement \nameref{rule:sprop} in a cheap
yet complete way, and this is the strategy used by \Coq.

\subparagraph{\(\eta\) and strengthening}
In the example of \cref{sec:algorithms},
the type-directed algorithm introduces a fresh variable to the context while the
untyped one does not.
This discrepancy between the algorithms explains the need for strengthening
in \cref{sec:meta-theory,prop:typed-to-untyped}.

While we do not have an example of a type theory that would stress
the equivalence by breaking strengthening, remark that in extensional type theory, where any two
terms are convertible in an inconsistent context and strengthening fails, to deduce
\(\conv{x,y : \Empty \to \nat}{x}{y}[\Empty \to \nat]\) it is necessary to go through
\(\eta\)-expansion to introduce a variable \(z : \Empty\) in the context.
Extensional type theory is undecidable,
thus less interesting to us, but a similar phenomenon happens in cubical type theory.
In \textsc{Cubical Agda}, elements of the type \texttt{Partial i0 nat}
behave as functions \texttt{isOne i0 -> nat}, but are compared
only on the geometric condition encoded by \texttt{isOne i0}.
Since this condition is degenerate, any two variables of type \texttt{Partial i0 -> nat}
are convertible, but only through a form of \(\eta\)-expansion.
In other words, strengthening fails with respect to a variable of type \texttt{isOne i0}.

For this and similar reasons, as far as we can tell it is widely believed that
cubical systems cannot be implemented using a purely term-directed conversion checker,
and given this hard failure of strengthening and neutral completeness, we are inclined
to follow this belief.

\section{Conclusion, related and future work}
\label{sec:conclusion}
\label{sec:related}

\subparagraph{Similar formalisations}
We are neither the only, nor the first, to work on formalising the meta-theory of
a dependent type system. These formalisations roughly fall into two groups.
The first \cite{Abel2017,Wieczorek2018,Hu2023,Adjedj2024,Jang2025,Liu2025}
tackles systems roughly the same complexity as ours,
and focus on meta-theory via varying forms of logical relations.
They all prove everything all they need in one go via the logical relation, and thus
do not attempt a detailed analysis of dependency between properties.

A second group focuses on describing and specifying realistic
type checkers. \Leanf's \cite{Carneiro2024} checker is on par with \Lean's kernel,
but only describes its intended type system without relating it to the checker,
and develops only minimal meta-theory.
\Agdac \cite{Liesnikov2025} defines a complex system modelled after \Agda,
and provide a checker returning typing derivations, thus ensuring
positive soundness by construction. They also do not consider meta-theory or properties
beyond positive soundness.
Only \MetaCoq \cite{MetaCoq2024} develops a comprehensive meta-theory,
for a type theory close to \Coq's, backing a certified sound,
complete and terminating checker ---up to
normalisation, which is axiomatised. We hope that the present work
clarifies the role of meta-theory in these projects: we encourage them
to adopt the ``meta-theory as black box'' motto, and focus on kernel certification
against a handful of axiomatised properties expected to hold ---similar to \MetaCoq's approach
to normalisation.

\subparagraph{Avoiding normalisation}
A wealth of approaches exist to establish injectivity properties independently of normalisation.
Confluence is powerful and scalable, as demonstrated by \MetaCoq~\cite{MetaCoq2024}.
Parallel reduction has been adapted to
typed conversion \cite{Adams2006,Siles2012} and \(\eta\) for
functions and pairs~\cite{Stoevring2006}, although neither work
tackle the combination of the two.
%
A different approach is taken by Coquand and Huber \cite{Coquand2018a},
who use a semantic in domains to obtain rich injectivity properties for
a type theory very close to ours.
The key point is that they carry out their construction in a weak ambient theory,
although they rely on the stratification of universes,
which shows that the logical power of injectivity properties is low.

Another approach to avoid normalisation is that of Abel and Altenkirch
\cite{Abel2011}, who show a coinductive form of completeness which makes no promise
about termination. They also put forward confluence, injectivity of type constructors,
and \emph{standardisation}, which gives a form of completeness
of weak-head reduction: if a term is convertible to a weak-head normal form, it
(weak-head) reduces to a weak-head normal form with the same shape.
This also appeared naturally in our formalisation, although less prominently than injectivity.

\subparagraph{Intrinsic typing}
An alternative approach to meta-theory \cite{Sterling2021a,Bocquet2023}
is based on ``intrinsic'' typing, where well-typed terms are defined
directly as an initial algebra or quotient inductive-inductive type (QIIT) \cite{Kaposi2019a},
letting one use powerful semantic tools to prove syntactic properties.
We would like to yield this powerful hammer, yet important hurdles remain.

In the intrinsic approach, syntax is by construction quotiented by conversion,
which makes it difficult to draw finer-grained distinctions.
Thus, while expressing injectivity of type constructors is straightforward,
neutral comparison, and thus \cref{prop:neu-compl,prop:neu-compl-pos},
is more problematic. Accordingly, these approaches rephrase decidability
in a ``reduction-free'' manner. 
It seems difficult to argue about termination of our function
based only on such conversion-invariant properties, we likely need something more intensional. 
A way out might be to reason about conversion proofs
as elements of the meta-level identity type of object terms, although this is impossible
with current QIITs, which are Set-truncated. 
A final issue is that we ultimately wish to certify an implementation acting on raw
terms, rather than intrinsic syntax. Relating the two amounts to the initiality
theorem \cite{Lumsdaine2020}, which is no small feat!

\subparagraph{Future work}
Given the centrality of injectivity properties, a natural future direction
is to develop new techniques focused on them,
especially ones that would scale to the complexity of real kernels. Even if
a normalisation proof for \Coq, \Lean or \Agda is a major endeavour not likely to manifest soon
---and will always have to circumvent Gödelian limitations---,
this seems more reachable, although we are not here yet. Whether confluence-like techniques
will suffice, or whether we will need more semantic approaches, remains to be explored.

On the other side of implementation,
we have started a minimal exploration of the world of conversions, by showing
we can tackle untyped conversion even against a typed specification. Our implementation
is rather naïve, and \textsc{PartialFun} is currently not geared towards good extraction.
Exploring the certification of efficient, optimised conversion checkers, inspired by those
explored ---and certified!--- by Courant \cite{Courant2024}, or the
work of Kovács on normalisation by evaluation \cite{Kovacs2024} is also an interesting
aspiration.



\bibliography{biblio}

\appendix


\section{Declarative Martin-Löf Type Theory}
\label{sec:complete-decl-tt}

On top of \(\sub{t}{u}\) for the substitution of the last variable of \(t\) by \(u\), we
more generally use \(\sub{t}{u,v,w}\) for the parallel
substitution of the last variables of \(t\) with \(u\), \(v\) and \(w\).

\subsection{Declarative typing}
\label{app:decl-typing}

\begin{mathparpagebreakable}
  \jform{\ctxtyping{\Gamma}}[Context $\Gamma$ is well-formed]
  \inferdef{ }{\ctxtyping{\cdot}} \and
  \inferdef{\ctxtyping{\Gamma} \\ \typing{\Gamma}{A}[\univ]}
      {\ctxtyping{\Gamma, x : A}}
\end{mathparpagebreakable}
  
\begin{mathparpagebreakable}
  \jform{\typing{\Gamma}{\sigma}[\Delta]}[$\sigma$ is a well-typed substitution between contexts $\Gamma$ and $\Delta$]
  \inferdef{ }{\typing{\Gamma}{\cdot}[\cdot]} \and
  \inferdef{\typing{\Gamma}{\sigma}[\Delta] \\ \typing{\Gamma}{t}[\sub{A}{\sigma}]}
    {\typing{\Gamma}{(\sigma,t)}[\Delta,x : A]}
\end{mathparpagebreakable}
\begin{mathparpagebreakable}
  \jform{\typing{\Gamma}{T}}[Type $T$ is well-formed in context $\Gamma$]
  \inferdef[FunTy]
    {\typing{\Gamma}{A} \\ \typing{\Gamma, x : A}{B}}
    {\typing{\Gamma}{\P x : A . B}} \and
  \inferdef[SigTy]
    {\typing{\Gamma}{A} \\ \typing{\Gamma, x : A}{B}}
    {\typing{\Gamma}{\Sig x : A . B}} \and
  \inferdef[NatTy]
    {\ctxtyping{\Gamma}}
    {\typing{\Gamma}{\nat}} \and
  \inferdef[EmptyTy]
    {\ctxtyping{\Gamma}}
    {\typing{\Gamma}{\Empty}} \and
  \inferdef[IdTy]
    {\typing{\Gamma}{A}\\
      \typing{\Gamma}{a}[A]\\
      \typing{\Gamma}{a'}[A]}
    {\typing{\Gamma}{\Id(A,a,a')}} \and
  \inferdef[El]
    {\typing{\Gamma}{A}[\univ]}
    {\typing{\Gamma}{A}} \and
  \inferdef[UnivTy]
    {\ctxtyping{\Gamma}}
    {\typing{\Gamma}{\univ}}
\end{mathparpagebreakable}
\begin{mathparpagebreakable}
  \jform{\typing{\Gamma}{t}[T]}[Term $t$ has type $T$ under context $\Gamma$]
  \inferdef[Conv]
    {\typing{\Gamma}{t}[A] \\
      \conv{\Gamma}{A}{B}}
    {\typing{\Gamma}{t}[B]} \and
  \inferdef[Var]{
    \ctxtyping{\Gamma} \\ (x : A) \in \Gamma} 
    {\typing{\Gamma}{x}[A]} \and
  \inferdef[FunUni]
    {\typing{\Gamma}{A}[\univ] \\ \typing{\Gamma, x : A}{B}[\univ]}
    {\typing{\Gamma}{\P x : A . B}[\univ]} \and
  \inferdef[Abs]
    {\typing{\Gamma}{A} \\ \typing{\Gamma, x : A}{B} \\ \typing{\Gamma, x : A}{t}[B]}
    {\typing{\Gamma}{\l x : A . t}[\P x : A . B]}
    \and
  \inferdef[App]
    {\typing{\Gamma}{t}[\P x : A . B] \\ \typing{\Gamma}{u}[A]}
    {\typing{\Gamma}{t~u}[\sub{B}{u}]}
  \\
  \inferdef[SigUni]
    {\typing{\Gamma}{A}[\univ] \\ \typing{\Gamma, x : A}{B}[\univ]}
    {\typing{\Gamma}{\Sig x : A . B}[\univ]}
  \and
  \inferdef[Pair]
    {\typing{\Gamma}{t}[A] \\ \typing{\Gamma}{u}[\sub{B}{t}]}
    {\typing{\Gamma}{\pair[A.B]{t}{u}}[\Sig x : A . B]}
  \and
  \inferdef[Proj\(_1\)]
    {\typing{\Gamma}{p}[\Sig x : A . B]}
    {\typing{\Gamma}{\fst{p}}[A]}
  \and
  \inferdef[Proj\(_2\)]
    {\typing{\Gamma}{p}[\Sig x : A . B]}
    {\typing{\Gamma}{\snd{p}}[\sub{B}{\fst{p}}]}
  \\
  \inferdef[NatUni]
    {\ctxtyping{\Gamma}}
    {\typing{\Gamma}{\nat}[\univ]} \and
    \inferdef[Zero]
    {\ctxtyping{\Gamma}}
    {\typing{\Gamma}{0}[\nat]} \and
    \inferdef[Succ]
    {\typing{\Gamma}{n}[\nat]}
    {\typing{\Gamma}{\succ(n)}[\nat]} \and
    \inferdef[NatRec]{
      \typing{\Gamma}{n}[\nat] \\
      \typing{\Gamma,x : \nat}{P} \\
      \typing{\Gamma}{b_{0}}[\sub{P}{0}] \\
      \typing{\Gamma,x : \nat, y : \sub{P}{n}}{b_{\succ}}[\sub{P}{\succ(n)}]}
    {\typing{\Gamma}{\natelim{n}{x.P}{b_{0}}{x.y.b_{\succ}}}[\sub{P}{n}]} \and
  \inferdef[EmptyUni]
    {\ctxtyping{\Gamma}}
    {\typing{\Gamma}{\Empty}[\univ]} \and
  \inferdef[EmptyInd]{
    \typing{\Gamma}{e}[\Empty] \\
    \typing{\Gamma,x : \Empty}{P}
  }
  {\typing{\Gamma}{\eelim{e}{x.P}}[\sub{P}{e}]} \and
  %
  \inferdef[IdUni]{
    \typing{\Gamma}{A}\\
    \typing{\Gamma}{a}[A]\\
    \typing{\Gamma}{a'}[A]}
  {\typing{\Gamma}{\Id(A,a,a')}} \and
  \inferdef[ReflTm]{
    \typing{\Gamma}{A} \\
    \typing{\Gamma}{a}[A]}
    {\typing{\Gamma}{\refl[A][a]}[\Id(A,a,a)]} \and
  \inferdef[IdInd]{
    \typing{\Gamma}{A} \\
    \typing{\Gamma}{a}[A] \\
    \typing{\Gamma}{a'}[A] \\
    \typing{\Gamma}{e}[\Id(A,a,a')] \\
    \typing{\Gamma,x : A, y : \Id(A,a,y)}{P} \\
    \typing{\Gamma}{b}[\sub{P}{a,\refl[A][a]}]}
    {\typing{\Gamma}{\ielim[A][a]{e}{x.y.P}{b}}[\sub{P}{a',e}]}
\end{mathparpagebreakable}


\subsection{Declarative conversion}
\label{app:decl-conv}

\begin{mathparpagebreakable}
  \jform{\conv{\Gamma}{T}{T'}}[Types $T$ and $T'$ are convertible in context $\Gamma$]
  \inferdef[ReflTy]
    {\typing{\Gamma}{A}}
    {\conv{\Gamma}{A}{A}} \and
  \inferdef[SymTy]
    {\conv{\Gamma}{A}{B}}
    {\conv{\Gamma}{B}{A}} \and
  \inferdef[TransTy]
    {\conv{\Gamma}{A}{B} \\
    \conv{\Gamma}{B}{C}}
    {\conv{\Gamma}{A}{C}} \and
  \inferdef[ElC]
    {\conv{\Gamma}{A}{A'}[\univ]}
    {\conv{\Gamma}{A}{A'}} \and
  \inferdef[FunTyC]
    {\conv{\Gamma}{A}{A'} \\ \conv{\Gamma, x : A}{B}{B'}}
    {\conv{\Gamma}{\P x : A . B}{\P x : A'.B'}} \and
  \inferdef[SigTyC]
    {\conv{\Gamma}{A}{A'} \\ \conv{\Gamma, x : A}{B}{B'}}
    {\conv{\Gamma}{\Sig x : A . B}{\Sig x : A'.B'}} \and
  \inferdef[IdTyC]
    {\conv{\Gamma}{A}{A'} \\
      \conv{\Gamma}{t}{t'}[A]\\
      \conv{\Gamma}{u}{u'}[A]}
    {\conv{\Gamma}{\Id(A,t,u)}{\Id(A',t',u')}}
\end{mathparpagebreakable}

\begin{mathparpagebreakable}
  \jform{\conv{\Gamma}{t}{t'}[T]}[Terms $t$ and $t'$ are convertible at
    type $T$ in context $\Gamma$]
  \inferdef[Refl]
    {\typing{\Gamma}{t}[A]}
    {\conv{\Gamma}{t}{t}[A]} \and
  \inferdef[Sym]
    {\conv{\Gamma}{t}{u}[A]}
    {\conv{\Gamma}{u}{t}[A]} \and
  \inferdef[Trans]
    {\conv{\Gamma}{t}{u}[A] \and
    \conv{\Gamma}{u}{v}[A]}
    {\conv{\Gamma}{t}{v}[A]} \and
    \inferdef[Conv]
    {\conv{\Gamma}{t}{t'}[A] \\
      \conv{\Gamma}{A}{B}}
    {\conv{\Gamma}{t}{t'}[B]} \and
  \inferdef[FunCong]
    {\conv{\Gamma}{A}{A'}[\univ] \\ \conv{\Gamma, x : A}{B}{B'}[\univ]}
    {\conv{\Gamma}{\P x : A . B}{\P x : A'. B'}[\univ]}
  \and
  \text{other congruences omitted} \\
  \inferdef[\(\beta\)Fun]
    {
      \typing{\Gamma}{A} \\ \typing{\Gamma, x : A}{B} \\\\
      \typing{\Gamma, x : A}{t}[B]\\ \typing{\Gamma}{u}[A]}
    {\conv{\Gamma}{(\l x : A.t)\ u}{\sub{t}{u}}[\sub{B}{u}]} \and
  \inferdef[\(\eta\)Fun]
      {\typing{\Gamma}{f}[\P x : A.B]}
      {\conv{\Gamma}{f}{\l x : A.f\ x}[\P x : A . B]}
        \\
  \inferdef[\(\beta\)Sig\(_1\)]
    {\typing{\Gamma}{A} \\ \typing{\Gamma, x : A}{B}\\\\
      \typing{\Gamma}{t}[A] \\
      \typing{\Gamma}{u}[\sub{B}{t}]}
    {\conv{\Gamma}{\fst{\pair[A.B]{t}{u}}}{t}[A]}
  \and
  \inferdef[\(\beta\)Sig\(_2\)]
    {\typing{\Gamma}{A} \\ \typing{\Gamma, x : A}{B}\\\\
      \typing{\Gamma}{t}[A] \\
      \typing{\Gamma}{u}[\sub{B}{t}]}
    {\conv{\Gamma}{\snd{\pair[A.B]{t}{u}}}{u}[\sub{B}{t}]}
  \and
  \inferdef[\(\eta\)Sig]
    {\typing{\Gamma}{A} \\ \typing{\Gamma, x : A}{B}\\
      \typing{\Gamma}{p}[\Sig x : A . B]}
    {\conv{\Gamma}{p}{\pair[A.B]{\fst{p}}{\snd{p}}}[\Sig x : A . B]}
  \and
  \inferdef[\(\beta\)Zero]{
    \typing{\Gamma,x : \nat}{P} \\
    \typing{\Gamma}{b_{0}}[\sub{P}{0}] \\
    \typing{\Gamma,x : \nat, y : \sub{P}{x}}{b_{\succ}}[\sub{P}{\succ(x)}]}
  {\conv{\Gamma}{\natelim{0}{x.P}{b_{0}}{x.y.b_{\succ}}}{b_{0}}[\sub{P}{0}]}
  \and
  \inferdef[\(\beta\)Succ]{
    \typing{\Gamma}{n}[\nat] \\
    \typing{\Gamma,x : \nat}{P} \\
    \typing{\Gamma}{b_{0}}[\sub{P}{0}] \\
    \typing{\Gamma,x : \nat, y : \sub{P}{x}}{b_{\succ}}[\sub{P}{\succ(x)}]}
  {\conv{\Gamma}{\natelim{\succ(n)}{x.P}{b_{0}}{x.y.b_{\succ}}}{
    \sub{b_{\succ}}{n,\natelim{n}{x.P}{b_{0}}{x.y.b_{\succ}}}}[\sub{P}{\succ(n)}]}
  \and
  \inferdef[\(\beta\)Refl]{
    \typing{\Gamma}{A} \\
    \typing{\Gamma}{a}[A] \\
    \typing{\Gamma,x : A, y : \Id(A,a,x)}{P} \\
    \typing{\Gamma}{b}[\sub{P}{a,\refl[A][a]}]}
  {\conv{\Gamma}{\ielim[A][a]{\refl[A][a]}{x.y.P}{b}}
    {b}[\sub{P}{a,\refl[A][a]}]}
\end{mathparpagebreakable}

\subsection{Neutral comparison}
\label{app:neutral-comp}

\begin{mathparpagebreakable}
  \jform{\nconv{\Gamma}{n}{n'}{T}}[
    Neutrals \(n\) and \(n'\) are convertible at type \(T\)]
  \inferdef[NConv]
    {\nconv{\Gamma}{n}{n'}{T} \\ \conv{\Gamma}{T}{S}}
    {\nconv{\Gamma}{n}{n'}{S}} \and
  \inferdef[NVar]
    {(x : T) \in \Gamma}
    {\nconv{\Gamma}{x}{x}{T}}
  \and
  \inferdef[NApp]
    {\nconv{\Gamma}{n}{n'}{\P x : A . B} \\ \conv{\Gamma}{u}{u'}[A]}
    {\nconv{\Gamma}{n\ u}{n'\ u'}{\sub{B}{u}}}
  \and
  \inferdef[NSig\(_{1}\)]
    {\nconv{\Gamma}{n}{n'}{\Sig x : A.B}}
    {\nconv{\Gamma}{\fst{n}}{\fst{n'}}{A}}
  \and
  \inferdef[NSig\(_{2}\)]
    {\nconv{\Gamma}{n}{n'}{\Sig x : A.B}}
    {\nconv{\Gamma}{\snd{n}}{\snd{n'}}{\sub{B}{\fst{n}}}}
  \and
  \inferdef[NNatElim]{
    \nconv{\Gamma}{n}{n'}{\nat} \\
    \conv{\Gamma,x : \nat}{P}{P'} \\
    \conv{\Gamma}{b_{0}}{b'_{0}}[\sub{P}{0}] \\
    \conv{\Gamma,x : \nat, y : \sub{P}{n}}{b_{\succ}}{b'_{\succ}}[\sub{P}{\succ(n)}]}
    {\nconv{\Gamma}{\natelim{n}{x.P}{b_{0}}{x.y.b_{\succ}}}{\natelim{n'}{x.P'}{b'_{0}}{x.y.b'_{\succ}}}{\sub{P}{n}}} \and
    \inferdef[NEmptyElim]{
      \nconv{\Gamma}{n}{n'}{\Empty} \\
      \conv{\Gamma,x : \Empty}{P}{P'}
    }{
      \nconv{\Gamma}{\eelim{n}{x.P}}{\eelim{n'}{x.P'}}{\sub{P}{s}}
    } \and
    \inferdef[NIdInd]{
      \nconv{\Gamma}{n}{n'}{\Id(A'',a'',b'')} \\
      \conv{\Gamma,x : A, y : \Id(A,a,x)}{P}{P'} \\
      \conv{\Gamma}{h}{h'}[\sub{P}{a,\refl[A][a]}]}
      {\nconv{\Gamma}{\ielim[A][a]{n}{x.y.P}{h}}
        {\ielim[A'][a']{n'}{x.y.P'}{h'}}{\sub{P}{b''}}}
\end{mathparpagebreakable}

\subsection{Deep normalisation}
\label{sec:deep-norm}

\begin{mathparpagebreakable}
  \jform{\dnf{\Gamma}{T}}[Type \(T\) is deeply normalising]
  \inferdef
  {T \red U \\
    \dnf*{\Gamma}{U}}
  {\dnf{\Gamma}{T}}
  \\
  \jform{\dnf{\Gamma}{t}[A]}[Term \(t\) is deeply normalising at type \(A\)]
  \inferdef
  {T \red U \\ t \red u \\
    \dnf*{\Gamma}{u}[U]}
  {\dnf{\Gamma}{t}[T]}
\end{mathparpagebreakable}

\begin{mathparpagebreakable}
  \jform{\dnf*{\Gamma}{T}}[Type \(T\) is a deeply normalising normal type]
  \inferdef
    {\dnf{\Gamma}{A} \\
    \dnf{\Gamma, x : A}{B}}
    {\dnf*{\Gamma}{\P x : A . B}}
  \and
  \inferdef
    {\dnf{\Gamma}{A} \\
    \dnf{\Gamma, x : A}{B}}
    {\dnf*{\Gamma}{\Sig x : A . B}}
  \and
  \inferdef
    { }
    {\dnf*{\Gamma}{\nat}}
  \and
  \inferdef
    { }
    {\dnf*{\Gamma}{\Empty}}
  \and
  \inferdef
    {\dnf{\Gamma}{A} \\
    \dnf{\Gamma}{t}[A] \\
    \dnf{\Gamma}{u}[A]}
    {\dnf*{\Gamma}{\Id(A,t,u)}} \and
  \inferdef
    { }
    {\dnf*{\Gamma}{\univ}}
  \and
  \inferdef
    {\dne{\Gamma}{n}{T}}
    {\dnf*{\Gamma}{n}}
\end{mathparpagebreakable}
\begin{mathparpagebreakable}
  \jform{\dnf*{\Gamma}{t}[A]}[
    Term \(t\) is a deeply normalising weak-head normal form at type \(A\)]
  \inferdef
  {\dnf{\Gamma}{A}[\univ] \\
    \dnf{\Gamma, x : A'}{B}[\univ]}
  {\dnf*{\Gamma}{\P x : A . B}[\univ]}
  \and 
  \inferdef
  {\dnf{\Gamma, x : A}{f~x}[B]}
  {\dnf*{\Gamma}{f}[\P x : A.B]}
  \and
  \inferdef
  {\dnf{\Gamma}{A}[\univ] \\
    \dnf{\Gamma, x : A'}{B}[\univ]}
  {\dnf*{\Gamma}{\Sig x : A . B}[\univ]}
  \and 
  \inferdef
  {\dnf{\Gamma}{\fst{p}}[A] \\
  \dnf{\Gamma}{\snd{p}}[\sub{B}{\fst{p}}]}
  {\dnf*{\Gamma}{p}[\Sig x : A.B]}
  \and
  \inferdef
    { }
    {\dnf*{\Gamma}{\nat}[\univ]}
  \and
  \inferdef
    { }
    {\dnf*{\Gamma}{0}[\nat]} \and
  \inferdef
    { \dnf{\Gamma}{t}[\nat]}
    {\dnf*{\Gamma}{\succ(t)}[\nat]} \and
  \inferdef
  { }
  {\dnf*{\Gamma}{\Empty}[\univ]} \and
  \inferdef
    {\dnf{\Gamma}{A}[\univ] \\
    \dnf{\Gamma}{t}[A] \\
    \dnf{\Gamma}{u}[A]}
    {\dnf*{\Gamma}{\Id(A,t,u)}[\univ]} \and
  \inferdef
    { }
    {\dnf{\Gamma}{\refl[A][a]}[\Id(A'',t,u)]} \and
  \inferdef
    {\dne{\Gamma}{n}{S} \\ \isPos T}
    {\dnf*{\Gamma}{n}[T]}
\end{mathparpagebreakable}

\begin{mathparpagebreakable}
  \jform{\dne*{\Gamma}{t}{T}}[
    Term \(t\) is a deeply normalising neutral at reduced type \(T\)]
  \inferdef
    {\dne{\Gamma}{n}{T} \\ T \red S}
    {\dne*{\Gamma}{n}{S}} \\
  \jform{\dne{\Gamma}{t}{T}}[
    Term \(t\) is a deeply normalising neutral at type \(T\)]
  \inferdef
    {(x : T) \in \Gamma}
    {\dne{\Gamma}{x}{T}}
  \and
  \inferdef
    {\dne*{\Gamma}{n}{\P x : A . B} \\ \dnf{\Gamma}{u}[A]}
    {\dne{\Gamma}{n~u}{\sub{B}{u}}}
  \and
  \inferdef
    {\dne*{\Gamma}{n}{\Sig x : A.B}}
    {\dne{\Gamma}{\fst{n}}{A}}
  \and
  \inferdef
    {\dne*{\Gamma}{n}{\Sig x : A.B}}
    {\dne{\Gamma}{\snd{n}}{\sub{B}{\fst{n}}}}
  \and
  \inferdef{
    \dne{\Gamma}{n}{\nat} \\
    \dnf{\Gamma,x : \nat}{P} \\
    \dnf{\Gamma}{b_{0}}[\sub{P}{0}] \\
    \dnf{\Gamma,x : \nat, y : \sub{P}{n}}{b_{\succ}}[\sub{P}{\succ(n)}]}
    {\dne{\Gamma}{\natelim{n}{x.P}{b_{0}}{x.y.b_{\succ}}}{\sub{P}{n}}} \and
    \inferdef{
      \dne*{\Gamma}{n}{\Empty} \\
      \dnf{\Gamma,x : \Empty}{P}
    }{
      \dne{\Gamma}{\eelim{n}{x.P}}{\sub{P}{s}}
    } \and
    \inferdef{
      \dne*{\Gamma}{n}{\Id(A'',a'',b'')} \\
      \dnf{\Gamma,x : A, y : \Id(A,a,x)}{P} \\
      \dnf{\Gamma}{h}[\sub{P}{a,\refl[A][a]}]}
      {\dne{\Gamma}{\ielim[A][a]{n}{x.y.P}{h}}{\sub{P}{b''}}}
\end{mathparpagebreakable}


\section{Algorithmic Martin-Löf Type Theory}
\label{sec:complete-algo-tt}

\subsection{Reduction}
\label{app:reduction}

\begin{mathparpagebreakable}
  \jform{t \ored t'}[Term $t$ weak-head reduces in one step to term $t'$]
  \inferdef[\(\beta\)Fun]
    { }
    {(\l x : A . t)\ u \ored \sub{t}{u}}
  \and
  \inferdef[\(\beta\)Sig\(_1\)]
    { }
    {\fst{\pair[A.B]{t}{u}} \ored t}
  \and
  \inferdef[\(\beta\)Sig\(_2\)]
    { }
    {\snd{\pair[A.B]{t}{u}} \ored u}
  \and
  \inferdef[\(\beta\)Zero]{ }{
    \natelim{0}{x.P}{b_{0}}{x.y.b_{\succ}} \ored b_{0}}
  \and
  \inferdef[\(\beta\)Succ]{ }{
    \natelim{\succ(n)}{x.P}{b_{0}}{x.y.b_{\succ}} \ored
    \sub{b_{\succ}}{n,\natelim{n}{x.P}{b_{0}}{x.y.b_{\succ}}}}
  \and
  \inferdef[\(\beta\)Refl]{ }
    {\ielim[A][a]{\refl[A][a]}{x.y.P}{b} \ored b}
  \\
  \inferdef[AppRed]
    {t \ored t'}
    {t~u \ored t'~u}
  \and
  \inferdef[AppFst]
    {t \ored t'}
    {\fst{t} \ored \fst{t'}}
  \and
  \inferdef[AppSnd]
    {t \ored t'}
    {\snd{t} \ored \snd{t'}}
  \and
  \inferdef[RedNatElim]
    {t \ored t'}
    {\natelim{t}{x.P}{b_{0}}{x.y.b_{\succ}}
      \ored \natelim{t'}{x.P}{b_{0}}{x.y.b_{\succ}}}
  \and
  \inferdef[RedIdElim]
  {t \ored t'}
  {\ielim[A][a]{t}{x.y.P}{b} \ored \ielim[A][a]{t'}{x.y.P}{b}}
\end{mathparpagebreakable}

\begin{mathpar}
  \jform{t \red t'}[Term $t$ weak-head reduces in multiple steps to term $t'$]
  \inferdef[RedBase]
  { }
  {t \red t}
  \and
  \inferdef[RedStep]
  {t \ored t' \\ t' \red t''}
  {t \red t''}
\end{mathpar}


\subsection{Normal forms}
\label{app:nf}

This BNF-style definitions of predicates on terms
should be read as \eg\ ``\(\can f\) is the predicate corresponding
to the sub-grammar of terms given by...'' The letter \(t\) stands for arbitrary terms,
and \(n\) is for neutrals.

\[
  \begin{array}{l c l r}
    \boxed{\can f} & \eqdef& \univ \mid \P x : t. t \mid  \l x : t. t \mid \nat \mid 0 \mid \succ(t) \mid
    & \text{\textcolor{gray}{weak-head canonical forms}}\\
    &&  \Sig x : t.t \mid \pair{t}{t} \mid \Empty \mid
    \Id(t,t,t) \mid \refl[t][t] \\\\
    \boxed{\neu n} & \eqdef& x \mid n\ t \mid \fst{n} \mid \snd{n} \mid
      \natelim{n}{x.t}{t}{x.y.t} \mid & \\
    && \eelim{n}{x.t} \mid \ielim[t][t]{n}{x.y.t}{t} & \text{\textcolor{gray}{weak-head neutrals}}\\\\
    \boxed{\nf f} & \eqdef & \can f \vee \neu f & \text{\textcolor{gray}{weak-head normal forms}}
    \\\\
    \boxed{\isTy f} & \eqdef & \univ \mid \P x : t. t \mid \Sig x : t.t \mid \nat \mid \Empty \mid \Id(t,t,t)
      \mid n & \text{\textcolor{gray}{Types in weak-head normal form}} \\\\
    \boxed{\isPos f} & \eqdef & \univ \mid  \nat \mid \Empty \mid \Id(t,t,t)
      \mid n & \text{\textcolor{gray}{Positive types in weak-head normal form}} \\\\
    \boxed{\isNat f} & \eqdef & 0 \mid \succ(t) \mid n &
      \text{\textcolor{gray}{Natural numbers in weak-head normal form}} \\\\
    \boxed{\isFun f} & \eqdef & \l x : t.t \mid n &
      \text{\textcolor{gray}{Natural numbers in weak-head normal form}} \\\\
    \boxed{\operatorname{isPair} f} & \eqdef & \pair{t}{t} \mid n &
      \text{\textcolor{gray}{Pair in weak-head normal form}} \\\\
    \boxed{\operatorname{isId} f} & \eqdef & \refl[t][t] \mid n &
      \text{\textcolor{gray}{Identity term in weak-head normal form}}
  \end{array}
\]


\subsection{Bidirectional typing}
\label{app:bidir-typing}

\begin{mathparpagebreakable}
  \jform{\checkty{\Gamma}{T}}[\(T\) is a type in \(\Gamma\)]
  \inferdef[Sort]{ }
    {\checkty{\Gamma}{\univ}}
    \and
  \inferdef[FunTy]
    {\checkty{\Gamma}{A} \\ \checkty{\Gamma, x : A}{B}}
    {\checkty{\Gamma}{\P x : A . B}} \and
  \inferdef[SigTy]
    {\checkty{\Gamma}{A} \\ \checkty{\Gamma, x : A}{B}}
    {\checkty{\Gamma}{\Sig x : A . B}}
  \and
  \inferdef[NatTy]
    { }
    {\checkty{\Gamma}{\nat}} \and
  \inferdef[EmptyTy]
    { }
    {\checkty{\Gamma}{\Empty}} \and
  \inferdef[IdTy]
    {\checkty{\Gamma}{A}\\
      \checkty{\Gamma}{a}[A]\\
      \checkty{\Gamma}{a'}[A]}
    {\checkty{\Gamma}{\Id(A,a,a')}} \and
  \inferdef[El]
    {\inferty*{\Gamma}{A}{\univ}}
    {\checkty{\Gamma}{A}}
\end{mathparpagebreakable}

\begin{mathparpagebreakable}
  \jform{\inferty{\Gamma}{t}{T}}[Term \(t\) infers type \(T\) in context \(\Gamma\)]
  \inferdef[Var]
    {(x : T) \in \Gamma}
    {\inferty{\Gamma}{x}{T}}
  \and
  \inferdef[Fun]
    {\inferty*{\Gamma}{A}{\univ} \\
      \inferty*{\Gamma, x : A}{B}[\univ]}
    {\inferty{\Gamma}{\P x : A . B}{\univ}}
  \and
  \inferdef[Abs]
    {\checkty{\Gamma}{A} \\ \inferty{\Gamma, x : A}{t}{B}}
    {\inferty{\Gamma}{\l x : A . t}{\P x : A . B}}
  \and
  \inferdef[App]
    {\inferty*{\Gamma}{t}{\P x : A . B} \\ \checkty{\Gamma}{u}[A]}
    {\inferty{\Gamma}{t~u}{\sub{B}{u}}}
  \and
  \inferdef[SigUniv]
  {\inferty*{\Gamma}{A}{\univ} \\
    \inferty*{\Gamma, x : A}{B}[\univ]}
  {\inferty{\Gamma}{\Sig x : A . B}{\univ}}
  \and
  \inferdef[Pair]
  {
    \inferty{\Gamma}{t}{A} \\
    \checkty{\Gamma, x : A}{B} \\
    \checkty{\Gamma}{u}[\sub{B}{t}]}
  {\inferty{\Gamma}{\pair[A.B]{t}{u}}{\Sig x : A . B}}
  \and
  \inferdef[Proj$_1$]
    {\inferty*{\Gamma}{p}{\Sig x : A . B}}
    {\inferty{\Gamma}{\fst{p}}{A}}
  \and
  \inferdef[Proj$_2$]
    {\inferty{\Gamma}{p}{\Sig x : A . B}}
    {\inferty{\Gamma}{\snd{p}}{\sub{B}{\fst{p}}}}
  \\
  \inferdef[NatUniv]
    { }
    {\inferty{\Gamma}{\nat}{\univ}}
  \and
  \inferdef[Zero]
    { }
    {\inferty{\Gamma}{0}{\nat}} \and
  \inferdef[Succ]
    {\checkty{\Gamma}{n}{\nat}}
    {\inferty{\Gamma}{\succ(n)}{\nat}} \and
  \inferdef[NatRec]{
    \checkty{\Gamma}{n}[\nat] \\
    \checkty{\Gamma,x : \nat}{P} \\
    \checkty{\Gamma}{b_{0}}[\sub{P}{0}] \\
    \checkty{\Gamma,x : \nat, y : \sub{P}{n}}{b_{\succ}}[\sub{P}{\succ(n)}]}
    {\inferty{\Gamma}{\natelim{n}{x.P}{b_{0}}{x.y.b_{\succ}}}[\sub{P}{n}]} \and
  \inferdef[Empty]
    { }
    {\inferty{\Gamma}{\Empty}{\univ}} \and
  \inferdef[EmptyInd]{
    \checkty{\Gamma}{s}[\Empty] \\
    \checkty{\Gamma,x : \Empty}{P}
  }
  {\inferty{\Gamma}{\eelim{s}{x.P}}{\sub{P}{s}}} \and
  \inferdef[IdTy]
  {\inferty*{\Gamma}{A}{\univ}\\
    \checkty{\Gamma}{a}[A]\\
    \checkty{\Gamma}{a'}[A]}
  {\inferty{\Gamma}{\Id(A,a,a')}{\univ}} \and
  \inferdef[ReflTm]
  {\checkty{\Gamma}{A} \\
    \checkty{\Gamma}{a}[A]}
  {\inferty{\Gamma}{\refl[A][a]}{\Id(A,a,a)}} \and
  \inferdef[IdInd]{
    \checkty{\Gamma}{A} \\
    \inferty*{\Gamma}{s}{\Id(A',a,a')} \\
    \checkty{\Gamma,x : A, y : \Id(A,a,x)}{P} \\
    \checkty{\Gamma}{b}[\sub{P}{a,\refl[A][x]}]}
  {\inferty{\Gamma}{\ielim[A][a]{s}{x.y.P}{b}}{\sub{P}{a,s}}} \\
  \jform{\checkty{\Gamma}{t}{T}}[Term \(t\) checks against type \(T\)]\\
  \inferdef[Check]
    {\inferty{\Gamma}{t}{T'} \\ \bconv{\Gamma}{T'}{T}}
    {\checkty{\Gamma}{t}{T}}
  \\
  \jform{\inferty*{\Gamma}{t}{T}}[Term \(t\) infers the reduced type \(T\)]
  \inferdef[InfRed]
    {\inferty{\Gamma}{t}{T} \\ T \red T'}
    {\inferty*{\Gamma}{t}{T'}}
\end{mathparpagebreakable}


\subsection{Typed algorithmic conversion}
\label{app:algo-typed-conv}

\begin{mathparpagebreakable}
  \jform{\bconv{\Gamma}{T}{T'}}[Types \(T\) and \(T'\) are convertible]
  \inferdef[TyRed]
  {T \red U \\ T' \red U' \\
    \bconv*{\Gamma}{U}{U'}}
  {\bconv{\Gamma}{T}{T'}}
  \\
  \jform{\bconv{\Gamma}{t}{t'}[A]}[Terms \(t\) and \(t'\) are convertible at type \(A\)]
  \inferdef[TmRed]
  {T \red U \\ t \red u \\ t' \red u' \\
    \bconv*{\Gamma}{u}{u'}[U]}
  {\bconv{\Gamma}{t}{t'}[T]}
\end{mathparpagebreakable}

\begin{mathparpagebreakable}
  \jform{\bconv*{\Gamma}{T}{T'}}[Reduced types \(T\) and \(T'\) are convertible]
  \inferdef[CProdTy]
    {\bconv{\Gamma}{A}{A'} \\
    \bconv{\Gamma, x : A'}{B}{B'}}
    {\bconv*{\Gamma}{\P x : A . B}{\P x : A'. B'}}
  \and
  \inferdef[CSigTy]
    {\bconv{\Gamma}{A}{A'} \\
    \bconv{\Gamma, x : A}{B}{B'}}
    {\bconv*{\Gamma}{\Sig x : A . B}{\Sig x : A'. B'}}
  \and
  \inferdef[CNatTy]
    { }
    {\bconv*{\Gamma}{\nat}{\nat}}
  \and
  \inferdef[CreflTy]
    { }
    {\bconv*{\Gamma}{\Empty}{\Empty}}
  \and
  \inferdef[CIdTy]
    {\bconv{\Gamma}{A}{A'} \\
    \bconv{\Gamma}{t}{t'}[A] \\
    \bconv{\Gamma}{u}{u'}[A]}
    {\bconv*{\Gamma}{\Id(A,t,u)}{\Id(A',t',u')}} \and
  \inferdef[CUniTy]
    { }
    {\bconv*{\Gamma}{\univ}{\univ}}
  \and
  \inferdef[NeuTy]
    {\bnconv{\Gamma}{n}{n'}{T}}
    {\bconv*{\Gamma}{n}{n'}} \\

  \jform{\bconv*{\Gamma}{t}{t'}[A]}[
    Reduced terms \(t\) and \(t'\) are convertible at type \(A\)]
  %
  \inferdef[CFun]
  {\bconv{\Gamma}{A}{A'}[\univ] \\
    \bconv{\Gamma, x : A'}{B}{B'}[\univ]}
  {\bconv*{\Gamma}{\P x : A . B}{\P x : A'. B'}[\univ]}
  \and 
  \inferdef[CFunEta]
  {\bconv{\Gamma, x : A}{f~x}{f'~x}[B]}
  {\bconv*{\Gamma}{f}{f'}[\P x : A.B]}
  \and
  \inferdef[CSig]
  {\bconv{\Gamma}{A}{A'}[\univ] \\
    \bconv{\Gamma, x : A'}{B}{B'}[\univ]}
  {\bconv*{\Gamma}{\Sig x : A . B}{\Sig x : A'. B'}[\univ]}
  \and 
  \inferdef[CSigEta]
  {\bconv{\Gamma}{\fst{p}}{\fst{p'}}[A] \\\\
  \bconv{\Gamma}{\snd{p}}{\snd{p'}}[\sub{B}{\fst{p}}]}
  {\bconv*{\Gamma}{p}{p'}[\Sig x : A.B]}
  \and
  \inferdef[CNat]
    { }
    {\bconv*{\Gamma}{\nat}{\nat}[\univ]}
  \and
  \inferdef[CZero]
    { }
    {\bconv*{\Gamma}{0}{0}[\nat]} \and
  \inferdef[CSucc]
    { \bconv{\Gamma}{t}{t'}[\nat]}
    {\bconv*{\Gamma}{\succ(t)}{\succ(t')}[\nat]} \and
  \inferdef[CEmpty]
  { }
  {\bconv*{\Gamma}{\Empty}{\Empty}[\univ]} \and
  \inferdef[CId]
    {\bconv{\Gamma}{A}{A'}[\univ] \\
    \bconv{\Gamma}{t}{t'}[A] \\
    \bconv{\Gamma}{u}{u'}[A]}
    {\bconv*{\Gamma}{\Id(A,t,u)}{\Id(A',t',u')}[\univ]} \and
  \inferdef[ReflRefl]
    { }
    {\bconv{\Gamma}{\refl[A][a]}{\refl[A'][a']}[\Id(A'',t,u)]} \and
  \inferdef[NeuPos]
    {\bnconv{\Gamma}{n}{n'}{S} \\ \isPos T}
    {\bconv*{\Gamma}{n}{n'}[T]}
\end{mathparpagebreakable}

\begin{mathparpagebreakable}
  \jform{\bnconv*{\Gamma}{t}{t'}{T}}[
    Neutrals \(t\) and \(t'\) are comparable, inferring the reduced type \(T\)]
  \inferdef[NRed]
    {\bnconv{\Gamma}{n}{n'}{T} \\ T \red S}
    {\bnconv*{\Gamma}{n}{n'}{S}} \\
  \jform{\bnconv{\Gamma}{t}{t'}{T}}[
    Neutrals \(t\) and \(t'\) are comparable, inferring the type \(T\)]
  \inferdef[NVar]
    {(x : T) \in \Gamma}
    {\bnconv{\Gamma}{x}{x}{T}}
  \and
  \inferdef[NApp]
    {\bnconv*{\Gamma}{n}{n'}{\P x : A . B} \\ \bconv{\Gamma}{u}{u'}[A]}
    {\bnconv{\Gamma}{n~u}{n'~u'}{\sub{B}{u}}}
  \and
  \inferdef[NSig\(_{1}\)]
    {\bnconv*{\Gamma}{n}{n'}{\Sig x : A.B}}
    {\bnconv{\Gamma}{\fst{n}}{\fst{n'}}{A}}
  \and
  \inferdef[NSig\(_{2}\)]
    {\bnconv*{\Gamma}{n}{n'}{\Sig x : A.B}}
    {\bnconv{\Gamma}{\snd{n}}{\snd{n'}}{\sub{B}{\fst{n}}}}
  \and
  \inferdef[NNatElim]{
    \bnconv{\Gamma}{n}{n'}{\nat} \\
    \bconv{\Gamma,x : \nat}{P}{P'} \\
    \bconv{\Gamma}{b_{0}}{b'_{0}}[\sub{P}{0}] \\
    \bconv{\Gamma,x : \nat, y : \sub{P}{n}}{b_{\succ}}{b'_{\succ}}[\sub{P}{\succ(n)}]}
    {\bnconv{\Gamma}{\natelim{n}{x.P}{b_{0}}{x.y.b_{\succ}}}{\natelim{n'}{x.P'}{b'_{0}}{x.y.b'_{\succ}}}{\sub{P}{n}}} \and
    \inferdef[NEmptyElim]{
      \bnconv*{\Gamma}{n}{n'}{\Empty} \\
      \bconv{\Gamma,x : \Empty}{P}{P'}
    }{
      \bnconv{\Gamma}{\eelim{n}{x.P}}{\eelim{n'}{x.P'}}{\sub{P}{s}}
    } \and
    \inferdef[NIdInd]{
      \bnconv*{\Gamma}{n}{n'}{\Id(A'',a'',b'')} \\
      \bconv{\Gamma,x : A, y : \Id(A,a,x)}{P}{P'} \\
      \bconv{\Gamma}{h}{h'}[\sub{P}{a,\refl[A][a]}]}
      {\bnconv{\Gamma}{\ielim[A][a]{n}{x.y.P}{h}}
        {\ielim[A'][a']{h'}{x.y.z.P'}{x.h'}}{\sub{P}{b''}}}
\end{mathparpagebreakable}

\subsection{Untyped algorithmic conversion}
\label{app:algo-untyped-conv}

\begin{mathparpagebreakable}
\jform{\uconv{t}{t'}}[Terms \(t\) and \(t'\) are convertible]
\inferdef[TmRed]
{t \red u \\ t' \red u' \\
  \uconv*{u}{u'}}
{\uconv{t}{t'}}

\jform{\uconv*{t}{t'}}[
  Reduced terms \(t\) and \(t'\) are convertible]

\inferdef[CUni]
  { }
  {\uconv*{\univ}{\univ}}
\and
\inferdef[CFun]
  {\uconv{A}{A'} \\ \uconv{B}{B'}}
  {\uconv*{\P x : A . B}{\P x : A'. B'}}
\and
\inferdef[CLam]
  {\uconv{t}{t'}}
  {\uconv*{\l x : A. t}{\l x : A'. t'}}
\and
\inferdef[CLamNe]
  {\neu n' \\ \uconv{t}{n'\ x}}
  {\uconv*{\l x : A. t}{n'}}
\and
\inferdef[CNeLam]
  {\neu n \\ \uconv{n\ x}{n'}}
  {\uconv*{n}{\l x : A'. t'}}
\and
\inferdef[CSig]
  {\uconv{A}{A'} \\
    \uconv{B}{B'}}
  {\uconv*{\Sig x : A . B}{\Sig x : A'. B'}}
\and
\inferdef[CPair]
  {\uconv{p}{p'} \\ \uconv{q}{q'}}
  {\uconv*{\pair{p}{q}}{\pair{p'}{q'}}}
\and
\inferdef[CPairNe]
  {\uconv{p}{\fst{n'}} \\ \uconv{q}{\snd{n'}}}
  {\uconv*{\pair{p}{q}}{n'}}
\and 
\inferdef[CNePair]
  {\uconv{\fst{n}}{p'} \\ \uconv{\snd{n}}{q'}}
  {\uconv*{n}{\pair{p'}{q'}}}
\and
\inferdef[CNat]
  { }
  {\uconv*{\nat}{\nat}}
\and
\inferdef[CZero]
  { }
  {\uconv*{0}{0}} \and
\inferdef[CSucc]
  {\uconv{t}{t'}}
  {\uconv*{\succ(t)}{\succ(t')}} \and
\inferdef[CEmpty]
  { }
  {\uconv*{\Empty}{\Empty}} \and
\inferdef[CId]
  {\uconv{A}{A'} \\
  \uconv{t}{t'} \\
  \uconv{u}{u'}}
  {\uconv*{\Id(A,t,u)}{\Id(A',t',u')}} \and
\inferdef[ReflRefl]
  { }
  {\uconv{\refl[A][a]}{\refl[A'][a']}} \and
\inferdef[NeuNeu]
  {\unconv{n}{n'}}
  {\uconv*{n}{n'}} \and
\jform{\unconv{n}{n'}}[
  Neutrals \(n\) and \(n'\) are comparable]
\inferdef[NVar]
  { }
  {\unconv{x}{x}}
\and
\inferdef[NApp]
  {\unconv{n}{n'} \\ \uconv{u}{u'}}
  {\unconv{n\ u}{n'\ u'}}
\and
\inferdef[NSig\(_{1}\)]
  {\unconv{n}{n'}}
  {\unconv{\fst{n}}{\fst{n'}}}
\and
\inferdef[NSig\(_{2}\)]
  {\unconv{n}{n'}}
  {\unconv{\snd{n}}{\snd{n'}}}
\and
\inferdef[NNatElim]{
  \unconv{n}{n'} \\
  \uconv{P}{P'} \\
  \uconv{b_{0}}{b'_{0}} \\
  \uconv{b_{\succ}}{b'_{\succ}}}
  {\unconv{\natelim{n}{x.P}{b_{0}}{x.y.b_{\succ}}}{\natelim{n'}{x.P'}{b'_{0}}{x.y.b'_{\succ}}}} \and
  \inferdef[NEmptyElim]{
    \unconv{n}{n'} \\
    \uconv{P}{P'}
  }{
    \unconv{\eelim{n}{x.P}}{\eelim{n'}{x.P'}}
  } \and
  \inferdef[NIdInd]{
    \unconv*{n}{n'} \\
    \uconv{P}{P'} \\
    \uconv{h}{h'}}
    {\unconv{\ielim[A][a]{n}{x.y.P}{h}}
      {\ielim[A'][a']{h'}{x.y.z.P'}{x.h'}}}
\end{mathparpagebreakable}

\end{document}